\documentclass[sigconf,nonacm]{acmart}

\usepackage[utf8]{inputenc}

\usepackage{mathtools}
\usepackage{enumitem}
\usepackage{makecell}

\usepackage{tikz}

\AtBeginDocument{}

\newcommand\blfootnote[1]{%
  \begingroup
  \renewcommand\thefootnote{}\footnote{#1}%
  \addtocounter{footnote}{-1}%
  \endgroup
}

\begin{document}

\title{Multiprocessor Scheduling with Memory Constraints: Fundamental Properties and Finding Optimal Solutions}

\author{P\'al Andr\'as Papp}
\orcid{0009-0005-6667-802X}
\email{pal.andras.papp@huawei.com}
\affiliation{
  \department{Computing Systems Lab}
  \institution{Huawei Zurich Research Center}
  \city{Zurich}
  \country{Switzerland}
}

\author{Toni B\"ohnlein}
\email{toni.boehnlein@huawei.com}
\orcid{0009-0001-2152-022X}
\affiliation{
  \department{Computing Systems Lab}
  \institution{Huawei Zurich Research Center}
  \city{Zurich}
  \country{Switzerland}
}

\author{Albert-Jan N. Yzelman}
\email{albertjan.yzelman@huawei.com}
\orcid{0000-0001-8842-3689}
\affiliation{
 \department{Computing Systems Lab}
  \institution{Huawei Zurich Research Center}
  \city{Zurich}
  \country{Switzerland}
}

\begin{abstract}
We study the problem of scheduling a general computational DAG on multiple processors in a 2-level memory hierarchy. This setting is a natural generalization of several prominent models in the literature, and it simultaneously captures workload balancing, communication, and data movement due to cache size limitations. We first analyze the fundamental properties of this problem from a theoretical perspective, such as its computational complexity. We also prove that optimizing parallelization and memory management separately, as done in many applications, can result in a solution that is a linear factor away from the optimum.

On the algorithmic side, we discuss a natural technique to represent and solve the problem as an Integer Linear Program (ILP). We develop a holistic scheduling algorithm based on this approach, and we experimentally study its performance and properties on a small benchmark of computational tasks. Our results confirm that the ILP-based method can indeed find considerably better solutions than a baseline which combines classical scheduling algorithms and memory management policies.
\end{abstract}

\ccsdesc[500]{Theory of computation~Parallel computing models}
\ccsdesc[300]{Theory of computation~Scheduling algorithms}
\ccsdesc[300]{Theory of computation~Integer programming}

\keywords{Multi-BSP, Scheduling, Red-blue pebble game, Limited memory, Integer Linear Programming}

\maketitle

\blfootnote{\copyright P\'al Andr\'as Papp, Toni B\"ohnlein and Albert-Jan N. Yzelman, 2025. This is the author's full version of the work, posted here for personal use. Not for redistribution. The definitive version was published in the 54th International Conference on Parallel Processing (ICPP 2025), https://doi.org/10.1145/3754598.3754676.}

\section{Introduction}

The efficient execution of complex computations is a fundamental problem in parallel computing, with countless applications ranging from scientific simulations to machine learning. In general, a static computational task is often modeled as a Directed Acyclic Graph (DAG), where the nodes represent the specific operations to execute, and the directed edges represent dependencies between them, i.e.\ the output of one operation is required as an input for another. These dependencies not only describe precedence constraints between the operations, but also imply data movement when the two nodes are assigned to different processors.

One of the most important aspects of this problem is to efficiently parallelize the execution on multiple processors. This is in itself a rather intricate task: we need to distribute the workload in a balanced way between the processors, and, at the same time, minimize the communication between them. These contradicting objectives often lead to a delicate trade-off. Finding the best parallel schedule for DAGs has been studied exhaustively, but for general DAGs, most of these works focus on very simple models in order to keep the problem tractable.

Another crucial aspect to optimize is the data movement between different levels of the memory hierarchy (vertical I/O). For instance, since the size of the cache is limited, we may need to repeatedly load values from RAM during the computation, which can be very costly. The prominent tool to capture this cost for general DAGs in a two-level memory hierarchy is the red-blue pebble game of Hong and Kung~\cite{hongkung}. However, due to the complexity of the problem, most works in this model focus on I/O lower bounds for concrete computational tasks.

When we aim to capture both parallel execution and memory limitations at the same time, this leads to an even more complex problem where finding the optimal schedule for general DAGs is remarkably challenging. While there are several models in the literature, like the Multi-BSP model of Valiant~\cite{multiBSP2}, which captures all of these aspects, most of the results in these models focus on the optimal scheduling of either a concrete (often very structured) computational DAG from a specific application, or a small subclass of DAGs at best. Hence, while the scheduling of general (irregular) workloads in these more realistic models is a fundamental problem, there is very little known about the theoretical properties of this problem or its optimal solutions.

The main goal of our work is to analyze the problem of efficiently scheduling a general computational DAG in a model that captures both parallel execution and memory constraints. As a first step, we define a natural new model for this problem that can be understood either as the restriction of the Multi-BSP model to $2$ levels on general DAGs, or as the generalization of multiprocessor red-blue pebbling to weighted DAGs. We refer to this DAG scheduling problem as MBSP scheduling, and we study its fundamental theoretical properties: we show that even separate parts of this problem are NP-hard, and that a synchronous and asynchronous interpretation of the model can lead to very different optimal schedules.

We then consider a natural two-stage approach towards scheduling, which first aims to find an efficient parallel schedule without the memory constraints, and then applies a cache management policy on this schedule to minimize the I/O cost. This separation is realistic in many application domains, where e.g.\ the scheduling and the memory management are optimized by different hardware/software components. We prove that in the worst case, this two-stage approach can return a schedule which is a linear factor away from the actual optimum, even if both stages are optimal separately. This indicates that for any theoretical guarantees, it is crucial to use a holistic approach to the entire problem.

On the algorithmic side, we present a formulation of the scheduling task as an Integer Linear Program (ILP). We discuss several techniques to make this approach more efficient on computational DAGs of moderate size, including a divide-and-conquer technique to split the problem into multiple smaller parts. We compare this ILP-based scheduler to a two-stage baseline that combines a recent multi-processor scheduling algorithm with a strong cache management policy. Our experiments confirm that the holistic ILP-based solution can indeed find significantly better schedules than the baseline in many cases: in our main experiments, the ILP-based scheduler obtains a $0.76\times$ factor geometric-mean cost reduction over the baseline. We also analyze how these improvements depend on different parameters of the model, and our experiments also provide interesting insights into some further aspects of DAG scheduling problems.

\subsubsection*{Paper organization}

After an overview of related work, Section~\ref{sec:model} defines our scheduling model. Section~\ref{sec:twostage} describes the two-stage approach and our main theorem. Section~\ref{sec:moretheory} adds further theoretical results. Section~\ref{sec:ilp} describes our ILP-based approach, and Section~\ref{sec:experiments} discusses our experimental results.

Note that in order to improve readability, many technical details are deferred to the appendix. In particular, Appendix~\ref{app:model} discusses the model definition in more detail, Appendix~\ref{app:proofs} contains the proof details for our theorems, Appendix~\ref{app:ilps} elaborates on the ILP methods, and Appendix~\ref{app:experiments} discusses the details of our experiments.

\section{Related Work}

The optimal parallel scheduling of computational DAGs has been studied in countless works and numerous different models since the 1960s. However, the vast majority of theoretical works on the topic consider very simple models, such as $(P|prec|C_{\text{max}})$, which is essentially equivalent to the PRAM model without any communication costs~\cite{DAG2proc1,DAG2proc2,weighted1,weighted3,approx1,approx2,approx3,approx4}, or $(P|prec;c|C_{\text{max}})$, where communication is only modeled by a fixed latency cost, allowing to communicate any amount of data in this time window~\cite{commDdef1,commDdef2,commDunit1,commDappR1,commDappR2,commDappR0,recent1,recent2,recent4}.

More sophisticated models, like BSP or LogP, are closely motivated by real-world computing architectures, and thus capture the actual scheduling cost much more accurately~\cite{BSPintro,LogP,BSPbook2,BSPbook1,Bsplib, BSPimpl1}. However, finding the optimal schedules in these models for general DAGs is very challenging. In these models, past research mostly focuses on schedules and lower bounds for specific computational DAGs from concrete applications, such as matrix multiplication~\cite{mccoll24memory,pebbling_matrices_hoefler}. There are only a few exceptions that consider general DAGs, which e.g.\ analyze the theoretical properties of BSP scheduling~\cite{BSP_DAG_opdas}, or experimentally evaluate heuristic schedulers in a custom BSP-like model~\cite{SPD}. Most of these models also do not consider memory limitations or I/O costs at all.

On the other hand, the red-blue pebble game captures the I/O cost of executing a computation in a two-level memory hierarchy, but only on a single processor~\cite{hongkung}. The results on this model include I/O lower bounds for several concrete computations from practice~\cite{elango2014characterizing,RBpebbling9,jain2020spectral}, as well as complexity and inapproximability results for general DAGs~\cite{demaine,pebbling_papp,pebbling_carpenter,MPP_opdas,RBpebbling8}. There are also several attempts to generalize this model to multiple processors or deeper hierarchies~\cite{MPP_opdas,pebbling_gleinig,pebbling_matrices_hoefler}, but this again leads to more complex models, so authors mostly focus on special subclasses of DAGs or concrete computational tasks when studying them. Besides the red-blue pebble game, there are also several alternative models that define or study I/O costs on specific subclasses of DAGs~\cite{new1,new2}.

The two closest models to ours in previous works are the Multi-BSP model of Valiant~\cite{multiBSP1, multiBSP2} and the multiprocessor red-blue pebble game of B\"ohnlein \textit{et al.}~\cite{MPP_opdas}. Multi-BSP is a generalization of BSP to architectures of arbitrary depth; it is a very broad model that captures all the aspects mentioned before. However, due to this generality, the few works on the Multi-BSP model all focus on the optimal execution of a highly structured computational DAG from a concrete application. In particular, even with only $2$ levels, to our knowledge, Multi-BSP has not been studied or even rigorously defined for general DAGs. On the other hand, multiprocessor red-blue pebbling is a special case of our model that is restricted to unit weights for both computation time and memory use, and to a synchronous setting. The study of this model in~\cite{MPP_opdas} only focuses on theoretical properties and complexity results, and does not consider algorithms to find a good schedule.

\section{Model definition} \label{sec:model}

Our model can be best described by combining the terminology of the Multi-BSP model and the red-blue pebble game. Similarly to red-blue pebbling, we assume a two-level memory hierarchy, where each processor has a fast memory (e.g.\ cache) of some limited capacity $r$, and the processors share a slow memory (e.g.\ RAM) of unlimited capacity. We use red and blue pebbles on a node $v$ to represent the fact that the output data of $v$ is currently kept in fast and slow memory, respectively. We have a different kind of red pebble for each processor (to denote the separate caches), but only a single kind of blue pebble (for the shared RAM).

The input of a concrete MBSP problem instance consists of:
\begin{itemize}[topsep=4pt, itemsep=3pt]
 \item a computational DAG $G=(V,E)$, where each node also has a compute weight $\omega:V \rightarrow \mathbb{R}_{\geq 0}$ (the time it takes to execute the operation), and a memory weight $\mu:V \rightarrow \mathbb{R}_{\geq 0}$ (the amount of memory its output requires);
 \item a computing architecture, consisting of a number of processors $P \in \mathbb{Z}_{> 0}$, a fast memory capacity $r \in \mathbb{R}_{\geq 0}$ that is identical for each processor, and the global parameters $g, L \in \mathbb{R}_{\geq 0}$ of the BSP model, where $g$ is the cost of sending a single unit of data, and $L$ is the cost of synchronization between the processors.
\end{itemize}
We denote the number of nodes in the DAG by $n\!=\!|V|$, and we refer to processors by numbering them from $1$ to $P$. In general, we will also use $[z]$ as a shorthand notation for the set of integers $\{1, ..., z\}$.

\subsection{Transition rules and pebbling sequence}

The current state of our schedule can described by a configuration $\zeta=(R_1, ..., R_{P}, B)$, where $R_p$ is the set of nodes with a red pebble of processor $p \! \in \! [P]$ (i.e.\ values currently in the cache of processor $p$), and $B$ is the set of nodes with a blue pebble (i.e.\ values in RAM). We require that every configuration in our schedule fulfills the memory bound, i.e.\ $\sum_{v \in R_p} \mu(v) \, \leq \, r$ for all $p \! \in \! [P]$.

In the initial configuration of our schedule, $B$ only contains the source nodes of the DAG (i.e., the inputs of the computation), and $R_p = \emptyset$ for all $p \! \in \! [P]$. By the terminal configuration, we need to have all the sink nodes of the DAG (i.e., the outputs) contained in $B$. During the schedule, we can use the following transition rules to move from the current to the next configuration:

\begin{enumerate}[topsep=4pt, itemsep=4pt]
\item LOAD$_{p,v}$: if node $v$ already has a blue pebble, then place a red pebble of processor $p$ on $v$. The cost of this rule is $cost(\text{LOAD}_{p,v}) = \mu(v) \cdot g$.
\item SAVE$_{p,v}$: if node $v$ already has a red pebble of processor $p$, then place a blue pebble on $v$. The cost of this rule is $cost(\text{SAVE}_{p,v}) = \mu(v) \cdot g$.
\item COMPUTE$_{p,v}$: if $v$ is not a source node, and all the parents of $v$ have a red pebble of processor $p$, then place a red pebble of processor $p$ also on $v$. The cost of this rule is $cost(\text{COMPUTE}_{p,v}) = \omega(v)$.
\item DELETE$_{p,v}$: remove a red pebble of processor $p$ from any node $v$. The cost of this rule is $cost(\text{DELETE}_{p,v}) = 0$.
\end{enumerate}

A pebbling sequence on processor $p$ is a sequence of configurations $(\zeta_0, \zeta_1, ..., \zeta_T)$, or equivalently, a sequence of transition rules $\Psi_p = (\tau_1, ..., \tau_T)$ from the list above, such that $\zeta_i$ can be obtained from $\zeta_{i-1}$ by using transition rule $\tau_i$, for all $i \! \in \! [T]$. Recall that all configurations $\zeta_i$ must fulfill the memory bound. The cost of the sequence is understood simply as
$\sum_{i \in [T]} cost(\tau_i)$.

\subsection{Supersteps and schedule}

Similarly to Multi-BSP, our schedule is organized into supersteps, each of which consists of a computation and a communication phase. This is a natural way to form schedules in a synchronous model like (Multi-)BSP. For our asynchronous model variant, we also define our schedules in the same fashion for the sake of simplicity; however, the splitting into computation and communication phases plays no role in this case.

Let us define a superstep on processor $p$ as a pebbling sequence $\Psi_p$ that can be obtained by concatenating four subsequences:
\[ \Psi_p = \Psi_{comp} \circ \Psi_{save} \circ \Psi_{del} \circ \Psi_{load} \, , \]
where
\begin{itemize}[topsep=3pt]
 \item $\Psi_{comp}$ only consists of COMPUTE$_{p,v}$ and DELETE$_{p,v}$ steps,
 \item $\Psi_{save}$ only consists of SAVE$_{p,v}$ steps,
 \item $\Psi_{del}$ only consists of DELETE$_{p,v}$ steps,
 \item $\Psi_{load}$ only consists of LOAD$_{p,v}$ steps.
\end{itemize}

Note that more in line with the BSP model, we could also simply assume $\Psi_p=\Psi_{comp} \circ \Psi_{comm}$, where $\Psi_{comp}$ only consists of compute and delete steps, and $\Psi_{comm}$ only consists of save, load and delete steps. However, such pebbling can always be reordered into a sequence of the form $\Psi_{comm} = \Psi_{save} \circ \Psi_{del} \circ \Psi_{load}$ without affecting its validity or increasing its cost.

A \emph{superstep $\Psi$} is then a tuple $\Psi=(\Psi_1, ..., \Psi_P)$, where $\Psi_p$ is a superstep on processor $p$ for $p \in [P]$. We point out that the formal definition here is actually a bit more technical, since the set $B$ is shared among the processors, and hence the pebbling sequences of different processors are not independent. However, the set $B$ is only modified during the $\Psi_{save}$ phase and only queried during the $\Psi_{load}$ phase. As such, we can formally define pebbling sequences using a separate artificial set $B_p$ for each processor, and then setting $B_p \leftarrow \bigcup_{p \in [P]} B_p$ at the end of the subsequence $\Psi_{save}$. We defer the details of this formal definition to the full version of the paper.

Finally, a MBSP \emph{schedule} is a sequence of supersteps
\[ S=(\Psi^{(1)}, \Psi^{(2)}, ..., \Psi^{(m)}) \] such that the finishing configuration of $\Psi^{(i)}$ is identical to the starting configuration of $\Psi^{(i_{\!}+_{\!}1)}$ for all $i \! \in \! [m_{\!}-_{\!}1]$, and furthermore, $\Psi^{(1)}$ starts with an initial configuration, and $\Psi^{(m)}$ ends with a terminal configuration.

\subsection{The cost of a schedule}

It only remains to define the cost function we aim to minimize in our problem. We discuss both a synchronous and an asynchronous interpretation of our schedules, which only differ in the cost function used to evaluate the total timespan of a schedule. The synchronous cost is very close to the (Multi-)BSP model in spirit: it considers the concrete superstep structure of our schedule, and also uses the synchronization parameter $L$. In contrast to this, the asynchronous cost is more similar to makespan metrics in classical scheduling models like $(P|prec|C_{\text{max}})$.

In the synchronous case, the cost of a superstep $\Psi$ can be simply defined as
\begin{equation*}
\begin{split}
cost(\Psi) = \max_{p \in [P]} \, cost(\Psi_{comp}) + \max_{p \in [P]} \, cost(\Psi_{save}) \, + \\ + \max_{p \in [P]} \, cost(\Psi_{load}) + L \, ,
\end{split}
\end{equation*}
and the cost of the whole schedule $S$ is simply the sum of these costs over all supersteps, i.e.\ $\sum_{i \in [m]} \, cost(\Psi^{(i)})$.

In the asynchronous case, we inductively define the finishing time $\gamma$ of each transition step in our schedule. Consider the entire sequence of transitions $(\tau_1, ...., \tau_T)$ on a given processor $p$ over all the supersteps. Let $\gamma(0) = 0$. For any transition $\tau_i$, we set
\begin{itemize}[topsep=3pt, itemsep=3pt]
 \item $\gamma(\tau_i) = \gamma(\tau_{i-1}) + cost(\tau_i)$, when $\tau_i$ is a save, compute or delete operation,
 \item $\gamma(\tau_i) = \max(\, \gamma(\tau_{i-1}) \, , \, \Gamma(v) \, ) + cost(\tau_i)$ when $\tau_i$ is a load operation,
\end{itemize}
where $\Gamma(v)$ is an auxiliary function to define the time when a node $v$ first becomes available in slow memory. More specifically, we consider the first superstep that has a SAVE$_{p,v}$ transition for $v$, and we define $\Gamma(v)$ as the finishing time $\gamma(\tau)$ of the first such transition step. This way, $\gamma(\tau_i)$ is indeed always well-defined in a valid schedule, and it expresses the time when $\tau_i$ is finished in an asynchronous execution of the schedule.

Finally, if we use $\tau_T\,\!^{(p)}$ to denote the last transition on processor $p$, then the asynchronous cost of the whole schedule $S$ is defined simply as $\max_{p \in [P]} \, \gamma(\tau_T\,\!^{(p)})$.

\section{The two-stage approach} \label{sec:twostage}

On a high level, MBSP scheduling can also be understood as the combination of two problems. Firstly, we need to find an efficient parallel schedule, i.e.\ decide which subtasks to execute on specific processors and time steps, in order to balance the workload but also minimize communication. Secondly, we need to develop a cache management policy on each processor, i.e.\ decide which values to load and evict throughout the schedule to satisfy the memory constraint, but also keep I/O costs low.

In general, these two subproblems are heavily interconnected: the memory constraints can also drastically influence the optimal way to develop the parallel schedule. However, in many practical cases, e.g.\ when the cache management policy of a system cannot be influenced directly, we have no other option than to handle these two stages separately. Since multiprocessor scheduling and red-blue pebbling are both rather challenging, it is also tempting from the algorithmic side to separate the two problems. That is, given an input problem, we can
\begin{itemize}[topsep=3pt]
 \item first find a good multi-processor schedule for the DAG, according to e.g.\ the BSP model, not considering the memory bound at this stage, and then
 \item find a good memory management policy for each of the processors in this schedule, to minimize I/O cost under the given memory constraint.
\end{itemize}
Our main goal is to compare this two-stage method to directly finding an optimal schedule for the MBSP problem as a whole. 

Note that it requires a further technical step to convert such a two-step solution to an MBSP schedule. In particular, the supersteps of the BSP schedule from the first stage may need to be further split up to obtain MBSP supersteps: in order to execute the sequence of computations in a given compute phase, we may need further I/O operations in-between to load new values from slow memory. In our implementations, we also use such a conversion: we form new supersteps for MBSP by splitting each BSP compute phase into maximally long segments of compute steps that can still be executed without a new I/O operation. The resulting compute phases ensure that a valid MBSP schedule exists; we can then combine this with a cache management policy in the second stage, always loading the new values needed for the next superstep, and evicting e.g.\ the least recently used values when required by the memory constraint.

\subsection{Proof of suboptimality}

As our main theoretical result, we show that this two-stage approach has strong limitations: in the worst case, the resulting cost is very far from the optimum.

\begin{theorem} \label{th:twostage}
The two-phase approach can return a solution that is a $\Theta(n)$ factor away from the optimum, even if both phases are optimal separately.
\end{theorem}

\renewcommand*{\proofname}{Proof sketch}

\begin{proof}
Consider a modified version of the construction from~\cite{pebbling_papp}. For some parameter $d=\Theta(n)$, we take two groups $H_1$, $H_2$ of $d$ source nodes each, and two chains of length $m$ each. The chain nodes also have incoming edges from the two groups $H_1$, $H_2$ in an alternating fashion, as shown in Figure~\ref{fig:zipper}. The construction allows us to set $m=\Theta(n)$ and $d=\Theta(n)$. Let the cache size be $r=d+2$, let us have $P=2$ processors, and $g=O(1)$, $L=0$. Our proof holds for both the synchronous and asynchronous cost model.

If we ensure that $d \! \cdot \! g < m$, then the optimal scheduling strategy in the BSP models is to assign a separate chain to both processors. Intuitively, this means that the processors can compute along their respective chains without any communication. This has a compute cost of $m$, and a communication cost of only $d \cdot g$ in the beginning to get the values of $H_1$ and $H_2$ to both processors (with some technical details depending on the concrete BSP variant). One can show that the alternative schedules indeed all yield a higher cost. If we use the two processors to compute all the children of $H_1$ and $H_2$, respectively, then every chain node needs to be sent to the other processor, giving a higher communication cost of $2\!\cdot \! m \! \cdot \! g$. If we compute both chains on the same processor, this increases the compute cost by $m$, so it is also suboptimal if $d \! \cdot \! g < m$. A more detailed case analysis is deferred to the full version of the paper.

Given this best BSP schedule from the first stage, if we set $r=d+2$, then any caching policy needs to repeatedly alternate between loading the values of $H_1$ and $H_2$ from slow memory, resulting in a total of $d \cdot m$ I/O operations in an MBSP schedule.

In contrast to this, an optimal MBSP schedule can assign all children of $H_1$ to one processor, and all children of $H_2$ to the other. The compute cost here is again $m$, but each chain node only incurs two I/O steps: the processors exchange the last computed chain node values by saving them into slow memory and then loading the value saved by the other processor. With that, the total I/O cost in the schedule is only $(2 \! \cdot \! m +d)\cdot g$. The two assignments of the chain nodes are also illustrated in Figure~\ref{fig:zipper_assign}.

\begin{figure}[t]
    \centering
    \begin{tikzpicture}

    \begin{scope}[thick]
    \draw (0pt,0pt) -- (15pt,14pt);
    \draw (0pt,23pt) -- (15pt,17.5pt);
    \draw (0pt,35pt) -- (15pt,21pt);

    \draw (0pt,65pt) -- (15pt,79pt);
    \draw (0pt,88pt) -- (15pt,82.5pt);
    \draw (0pt,100pt) -- (15pt,86pt);
    \end{scope}

    \begin{scope}[thick, arrows=-stealth]
    \draw (15pt,14pt) -- (145pt,14pt);
    \draw (15pt,17.5pt) -- (145pt,17.5pt);
    \draw (15pt,21pt) -- (145pt,21pt);

    \draw (15pt,79pt) -- (145pt,79pt);
    \draw (15pt,82.5pt) -- (145pt,82.5pt);
    \draw (15pt,86pt) -- (145pt,86pt);
    \end{scope}

    \draw[white, fill=white] (25.5pt,16.5pt) rectangle (28pt,18.5pt);
    \draw[white, fill=white] (55.5pt,16.5pt) rectangle (58pt,18.5pt);
    \draw[white, fill=white] (85.5pt,16.5pt) rectangle (88pt,18.5pt);
    \draw[white, fill=white] (115.5pt,16.5pt) rectangle (118pt,18.5pt);

    \draw[white, fill=white] (25.5pt,81.5pt) rectangle (28pt,83.5pt);
    \draw[white, fill=white] (55.5pt,81.5pt) rectangle (58pt,83.5pt);
    \draw[white, fill=white] (85.5pt,81.5pt) rectangle (88pt,83.5pt);
    \draw[white, fill=white] (115.5pt,81.5pt) rectangle (118pt,83.5pt);

    \draw[white, fill=white] (24.5pt,20pt) rectangle (31pt,22pt);
    \draw[white, fill=white] (54.5pt,20pt) rectangle (61pt,22pt);
    \draw[white, fill=white] (84.5pt,20pt) rectangle (91pt,22pt);
    \draw[white, fill=white] (114.5pt,20pt) rectangle (121pt,22pt);

    \draw[white, fill=white] (24.5pt,78pt) rectangle (31pt,80pt);
    \draw[white, fill=white] (54.5pt,78pt) rectangle (61pt,80pt);
    \draw[white, fill=white] (84.5pt,78pt) rectangle (91pt,80pt);
    \draw[white, fill=white] (114.5pt,78pt) rectangle (121pt,80pt);

    \begin{scope}[thick, densely dotted]
    \draw (25.5pt,17.5pt) -- (28pt,17.5pt);
    \draw (24.5pt,21pt) -- (31pt,21pt);
    \draw (55.5pt,17.5pt) -- (58pt,17.5pt);
    \draw (54.5pt,21pt) -- (61pt,21pt);
    \draw (85.5pt,17.5pt) -- (88pt,17.5pt);
    \draw (84.5pt,21pt) -- (91pt,21pt);
    \draw (115.5pt,17.5pt) -- (118pt,17.5pt);
    \draw (114.5pt,21pt) -- (121pt,21pt);

    \draw (25.5pt,82.5pt) -- (28pt,82.5pt);
    \draw (24.5pt,79pt) -- (31pt,79pt);
    \draw (55.5pt,82.5pt) -- (58pt,82.5pt);
    \draw (54.5pt,79pt) -- (61pt,79pt);
    \draw (85.5pt,82.5pt) -- (88pt,82.5pt);
    \draw (84.5pt,79pt) -- (91pt,79pt);
    \draw (115.5pt,82.5pt) -- (118pt,82.5pt);
    \draw (114.5pt,79pt) -- (121pt,79pt);
    \end{scope}

    \begin{scope}[thick, arrows=-stealth]

    \draw (24pt,14pt) -- (41pt,36pt);
    \draw (23pt,17.5pt) -- (38pt,37pt);
    \draw (22pt,21pt) -- (36pt,39pt);

    \draw (54pt,14pt) -- (71pt,36pt);
    \draw (53pt,17.5pt) -- (68pt,37pt);
    \draw (52pt,21pt) -- (66pt,39pt);

    \draw (84pt,14pt) -- (101pt,36pt);
    \draw (83pt,17.5pt) -- (98pt,37pt);
    \draw (82pt,21pt) -- (96pt,39pt);

    \draw (114pt,14pt) -- (131pt,36pt);
    \draw (113pt,17.5pt) -- (128pt,37pt);
    \draw (112pt,21pt) -- (126pt,39pt);

    \draw (24pt,86pt) -- (41pt,64pt);
    \draw (23pt,82.5pt) -- (38pt,63pt);
    \draw (22pt,79pt) -- (36pt,61pt);

    \draw (54pt,86pt) -- (71pt,64pt);
    \draw (53pt,82.5pt) -- (68pt,63pt);
    \draw (52pt,79pt) -- (66pt,61pt);

    \draw (84pt,86pt) -- (101pt,64pt);
    \draw (83pt,82.5pt) -- (98pt,63pt);
    \draw (82pt,79pt) -- (96pt,61pt);

    \draw (114pt,86pt) -- (131pt,64pt);
    \draw (113pt,82.5pt) -- (128pt,63pt);
    \draw (112pt,79pt) -- (126pt,61pt);

    \draw (40pt,40pt) -- (67pt,57pt);
    \draw (70pt,40pt) -- (97pt,57pt);
    \draw (100pt,40pt) -- (127pt,57pt);
    \end{scope}

    \draw[white, fill=white] (54.5pt,48pt) rectangle (57.5pt,52pt);
    \draw[white, fill=white] (84.5pt,48pt) rectangle (87.5pt,52pt);
    \draw[white, fill=white] (114.5pt,48pt) rectangle (117.5pt,52pt);

    \begin{scope}[thick, arrows=-stealth]
    \draw (40pt,60pt) -- (67pt,43pt);
    \draw (70pt,60pt) -- (97pt,43pt);
    \draw (100pt,60pt) -- (127pt,43pt);
    \end{scope}
    
    \begin{scope}[thick]
    \draw (130pt,40pt) -- (140pt,46pt);
    \draw (130pt,60pt) -- (140pt,54pt);
    \end{scope}

    \draw[black, fill=white] (0pt,0pt) circle (1.0ex);
    \node[anchor=center] at (0pt,11pt) {\large $...$};
    \draw[black, fill=white] (0pt,23pt) circle (1.0ex);
    \draw[black, fill=white] (0pt,35pt) circle (1.0ex);

    \draw[black, fill=white] (0pt,65pt) circle (1.0ex);
    \node[anchor=center] at (0pt,76pt) {\large $...$};
    \draw[black, fill=white] (0pt,88pt) circle (1.0ex);
    \draw[black, fill=white] (0pt,100pt) circle (1.0ex);

    \draw[black, fill=white] (40pt,40pt) circle (1.0ex);
    \draw[black, fill=white] (40pt,60pt) circle (1.0ex);
    \draw[black, fill=white] (70pt,40pt) circle (1.0ex);
    \draw[black, fill=white] (70pt,60pt) circle (1.0ex);
    \draw[black, fill=white] (100pt,40pt) circle (1.0ex);
    \draw[black, fill=white] (100pt,60pt) circle (1.0ex);
    \draw[black, fill=white] (130pt,40pt) circle (1.0ex);
    \draw[black, fill=white] (130pt,60pt) circle (1.0ex);

    \begin{scope}[very thick, gray]
    \draw (-6pt,-5pt) -- (-9pt,-5pt) -- (-9pt,40pt) -- (-6pt,40pt);
    \draw (-9pt,17.5pt) -- (-12pt,17.5pt);

    \draw (-6pt,60pt) -- (-9pt,60pt) -- (-9pt,105pt) -- (-6pt,105pt);
    \draw (-9pt,82.5pt) -- (-12pt,82.5pt);
    \end{scope}

    \node[anchor=center, rotate=90] at (-19pt,17.5pt) {\small $d$ nodes};
    \node[anchor=center, rotate=90] at (-19pt,82.5pt) {\small $d$ nodes};

\end{tikzpicture}
    \caption{Example construction for Theorem~\ref{th:twostage} where the two-stage approach returns a schedule of significantly higher cost than the actual optimum.}
    \label{fig:zipper}
\end{figure}
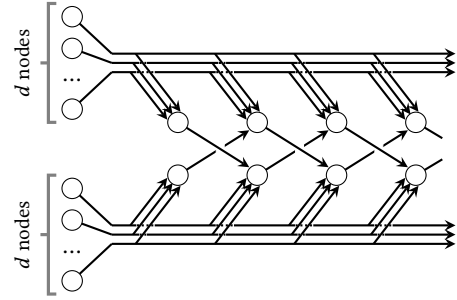

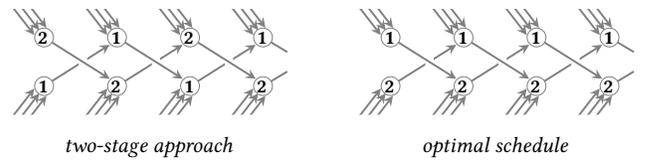
\begin{figure}[t]
    \centering
    \hspace{-0.03\textwidth}
        \begin{minipage}{0.22\textwidth}
            \resizebox{1.0\textwidth}{!}{\begin{tikzpicture}

    \begin{scope}[thick, arrows=-stealth, gray]

    \draw (24pt,14pt) -- (41pt,36pt);
    \draw (23pt,17.5pt) -- (38pt,37pt);
    \draw (22pt,21pt) -- (36pt,39pt);

    \draw (54pt,14pt) -- (71pt,36pt);
    \draw (53pt,17.5pt) -- (68pt,37pt);
    \draw (52pt,21pt) -- (66pt,39pt);

    \draw (84pt,14pt) -- (101pt,36pt);
    \draw (83pt,17.5pt) -- (98pt,37pt);
    \draw (82pt,21pt) -- (96pt,39pt);

    \draw (114pt,14pt) -- (131pt,36pt);
    \draw (113pt,17.5pt) -- (128pt,37pt);
    \draw (112pt,21pt) -- (126pt,39pt);

    \draw (24pt,86pt) -- (41pt,64pt);
    \draw (23pt,82.5pt) -- (38pt,63pt);
    \draw (22pt,79pt) -- (36pt,61pt);

    \draw (54pt,86pt) -- (71pt,64pt);
    \draw (53pt,82.5pt) -- (68pt,63pt);
    \draw (52pt,79pt) -- (66pt,61pt);

    \draw (84pt,86pt) -- (101pt,64pt);
    \draw (83pt,82.5pt) -- (98pt,63pt);
    \draw (82pt,79pt) -- (96pt,61pt);

    \draw (114pt,86pt) -- (131pt,64pt);
    \draw (113pt,82.5pt) -- (128pt,63pt);
    \draw (112pt,79pt) -- (126pt,61pt);

    \draw (40pt,40pt) -- (67pt,57pt);
    \draw (70pt,40pt) -- (97pt,57pt);
    \draw (100pt,40pt) -- (127pt,57pt);
    \end{scope}

    \draw[white, fill=white] (54.5pt,48pt) rectangle (57.5pt,52pt);
    \draw[white, fill=white] (84.5pt,48pt) rectangle (87.5pt,52pt);
    \draw[white, fill=white] (114.5pt,48pt) rectangle (117.5pt,52pt);

    \begin{scope}[thick, arrows=-stealth, gray]
    \draw (40pt,60pt) -- (67pt,43pt);
    \draw (70pt,60pt) -- (97pt,43pt);
    \draw (100pt,60pt) -- (127pt,43pt);
    \end{scope}
    
    \begin{scope}[thick, gray]
    \draw (130pt,40pt) -- (140pt,46pt);
    \draw (130pt,60pt) -- (140pt,54pt);
    \end{scope}

    \draw[gray, fill=white] (40pt,40pt) circle (1.0ex);
    \draw[gray, fill=white] (40pt,60pt) circle (1.0ex);
    \draw[gray, fill=white] (70pt,40pt) circle (1.0ex);
    \draw[gray, fill=white] (70pt,60pt) circle (1.0ex);
    \draw[gray, fill=white] (100pt,40pt) circle (1.0ex);
    \draw[gray, fill=white] (100pt,60pt) circle (1.0ex);
    \draw[gray, fill=white] (130pt,40pt) circle (1.0ex);
    \draw[gray, fill=white] (130pt,60pt) circle (1.0ex);

    \draw[white, fill=white] (22pt,28pt) rectangle (132pt,13pt);
    \draw[white, fill=white] (22pt,72pt) rectangle (132pt,87pt);

    \node[anchor=center] at (39.5pt,40pt) {\small \textbf{1}};
    \node[anchor=center] at (69.5pt,60pt) {\small \textbf{1}};
    \node[anchor=center] at (99.5pt,40pt) {\small \textbf{1}};
    \node[anchor=center] at (129.5pt,60pt) {\small \textbf{1}};

    \node[anchor=center] at (39.5pt,60pt) {\small \textbf{2}};
    \node[anchor=center] at (69.5pt,40pt) {\small \textbf{2}};
    \node[anchor=center] at (99.5pt,60pt) {\small \textbf{2}};
    \node[anchor=center] at (129.5pt,40pt) {\small \textbf{2}};

    \node[anchor=center] at (83pt,15pt) {\textit{two-stage approach}};

\end{tikzpicture}}
        \end{minipage}
        \hspace{0.03\textwidth}
        \begin{minipage}{0.22\textwidth}
            \resizebox{1.0\textwidth}{!}
            {\begin{tikzpicture}

    \begin{scope}[thick, arrows=-stealth, gray]

    \draw (24pt,14pt) -- (41pt,36pt);
    \draw (23pt,17.5pt) -- (38pt,37pt);
    \draw (22pt,21pt) -- (36pt,39pt);

    \draw (54pt,14pt) -- (71pt,36pt);
    \draw (53pt,17.5pt) -- (68pt,37pt);
    \draw (52pt,21pt) -- (66pt,39pt);

    \draw (84pt,14pt) -- (101pt,36pt);
    \draw (83pt,17.5pt) -- (98pt,37pt);
    \draw (82pt,21pt) -- (96pt,39pt);

    \draw (114pt,14pt) -- (131pt,36pt);
    \draw (113pt,17.5pt) -- (128pt,37pt);
    \draw (112pt,21pt) -- (126pt,39pt);

    \draw (24pt,86pt) -- (41pt,64pt);
    \draw (23pt,82.5pt) -- (38pt,63pt);
    \draw (22pt,79pt) -- (36pt,61pt);

    \draw (54pt,86pt) -- (71pt,64pt);
    \draw (53pt,82.5pt) -- (68pt,63pt);
    \draw (52pt,79pt) -- (66pt,61pt);

    \draw (84pt,86pt) -- (101pt,64pt);
    \draw (83pt,82.5pt) -- (98pt,63pt);
    \draw (82pt,79pt) -- (96pt,61pt);

    \draw (114pt,86pt) -- (131pt,64pt);
    \draw (113pt,82.5pt) -- (128pt,63pt);
    \draw (112pt,79pt) -- (126pt,61pt);

    \draw (40pt,40pt) -- (67pt,57pt);
    \draw (70pt,40pt) -- (97pt,57pt);
    \draw (100pt,40pt) -- (127pt,57pt);
    \end{scope}

    \draw[white, fill=white] (54.5pt,48pt) rectangle (57.5pt,52pt);
    \draw[white, fill=white] (84.5pt,48pt) rectangle (87.5pt,52pt);
    \draw[white, fill=white] (114.5pt,48pt) rectangle (117.5pt,52pt);

    \begin{scope}[thick, arrows=-stealth, gray]
    \draw (40pt,60pt) -- (67pt,43pt);
    \draw (70pt,60pt) -- (97pt,43pt);
    \draw (100pt,60pt) -- (127pt,43pt);
    \end{scope}
    
    \begin{scope}[thick, gray]
    \draw (130pt,40pt) -- (140pt,46pt);
    \draw (130pt,60pt) -- (140pt,54pt);
    \end{scope}

    \draw[gray, fill=white] (40pt,40pt) circle (1.0ex);
    \draw[gray, fill=white] (40pt,60pt) circle (1.0ex);
    \draw[gray, fill=white] (70pt,40pt) circle (1.0ex);
    \draw[gray, fill=white] (70pt,60pt) circle (1.0ex);
    \draw[gray, fill=white] (100pt,40pt) circle (1.0ex);
    \draw[gray, fill=white] (100pt,60pt) circle (1.0ex);
    \draw[gray, fill=white] (130pt,40pt) circle (1.0ex);
    \draw[gray, fill=white] (130pt,60pt) circle (1.0ex);

    \draw[white, fill=white] (22pt,28pt) rectangle (132pt,13pt);
    \draw[white, fill=white] (22pt,72pt) rectangle (132pt,87pt);

    \node[anchor=center] at (39.5pt,60pt) {\small \textbf{1}};
    \node[anchor=center] at (69.5pt,60pt) {\small \textbf{1}};
    \node[anchor=center] at (99.5pt,60pt) {\small \textbf{1}};
    \node[anchor=center] at (129.5pt,60pt) {\small \textbf{1}};

    \node[anchor=center] at (39.5pt,40pt) {\small \textbf{2}};
    \node[anchor=center] at (69.5pt,40pt) {\small \textbf{2}};
    \node[anchor=center] at (99.5pt,40pt) {\small \textbf{2}};
    \node[anchor=center] at (129.5pt,40pt) {\small \textbf{2}};

    \node[anchor=center] at (83pt,15pt) {\textit{optimal schedule}};

\end{tikzpicture}}
        \end{minipage}
        \hspace{-0.02\textwidth}
    \caption{Illustration of the assignments of chain nodes to the two processors with the two-stage approach and the optimum schedule, in the construction of Figure~\ref{fig:zipper}.}
    \label{fig:zipper_assign}
\end{figure}

This means that altogether, the cost ratio between the two-stage approach and the optimum is
\[ \frac{m + (d \cdot m )\cdot g}{m + (2 \cdot m +d)\cdot g} \, . \]
We can lower bound the numerator by $d \! \cdot \! m \! \cdot \! g$, and upper bound the denominator by $4 \! \cdot \! m \! \cdot \! g$ if we have $m \! \geq \! d$ and $g \! \geq \! 1$, thus obtaining a ratio of at least $\frac{d}{4}$. With $\Theta(d)=\Theta(n)$ in the construction, this is indeed a linear factor difference.
\end{proof}

We note that the construction above has uniform node weights and $L=0$, so the result also carries over to the simpler model of multiprocessor red-blue pebbling~\cite{MPP_opdas}. The example is also easy to generalize to any constant number of processors $P \geq 3$. 

\section{Further Fundamental Properties} \label{sec:moretheory}

\subsection{Computational Complexity}

One natural question about our scheduling problem is its complexity. Unsurprisingly, finding the optimum is very challenging: since MBSP generalizes multi-processor pebbling, it follows from~\cite{MPP_opdas} that it is already NP-hard on simple classes of DAGs, and does not allow a polynomial-time approximation scheme.

As for the parts of the two-stage approach, the complexity of the scheduling stage has been thoroughly studied, and it is also known to be NP-hard~\cite{BSP_DAG_opdas}. However, surprisingly, we are not aware that the complexity of the memory management stage has been analyzed before. More specifically, assume that the compute steps in the phases $\Psi_{comp}$ of our MBSP schedule are already fixed for each processor and superstep, but we are still free to choose all the save, delete and load operations throughout the schedule. Our goal is to minimize the total (synchronous) cost of the $\Psi_{save}$ and $\Psi_{load}$ phases over the whole schedule.

Note that in the uniform-weight case when $\omega(v) = 1$ for all $v$, the best strategy is simply that whenever we need to evict a value from cache, we select the value which is not needed for the longest time (known as B\'el\'ady's optimal algorithm or Clairvoyant algorithm). This can easily be computed in polynomial time. However, with general node weights, the problem immediately becomes NP-hard.

\begin{lemma} \label{lem:weakNP}
With $P=1$, the memory management problem is weakly NP-hard.
\end{lemma}

\begin{lemma} \label{lem:strongNP}
With $P \geq 2$, the memory management problem is strongly NP-hard.
\end{lemma}

\begin{proof}
The lemmas can be shown with relatively standard reductions techniques from the number partitioning and $3$-partitioning problems, respectively.

On a high level, for Lemma~\ref{lem:weakNP}, given integers $\{a_1, ..., a_m\}$ that sum up to $\alpha$, we can create nodes $v_1, ..., v_m$ and $v'$ with memory weights $a_1, ..., a_m$ and $\frac{\alpha}{2}$, respectively, and ensure that first the nodes $v_1, ..., v_m$ are needed in cache, then only $v'$, and then $v_1, ..., v_m$ again. If there is a subset of $\{a_1, ..., a_m\}$ that sums up to exactly $\frac{\alpha}{2}$, then we can leave these nodes in cache when loading $v'$, and only reload nodes of total weight $\frac{\alpha}{2}$ afterwards. Otherwise, the total loading cost for the third computation is larger than $\frac{\alpha}{2}$.

For Lemma~\ref{lem:strongNP}, given integers $\{a_1, ..., a_{3\!\cdot\!m}\}$ that sum up to $\alpha \cdot  m$, we create nodes with memory weights $a_1, ..., a_{3 \cdot m}$, respectively, which are all needed in cache on processor $1$ by superstep $(m+1)$. We ensure that processor $2$ is always loading a weight of $\alpha$ in the first $m$ communication phases, and hence we can only avoid increasing the total costs if we can partition our weights into $m$ distinct subsets that each sum up to exactly $\alpha$.
\end{proof}

\subsection{Synchronous vs. Asynchronous Optimum}

Another interesting aspect is how the synchronous and asynchronous cost functions relate to each other. The two costs can only be fairly compared when $L=0$, i.e.\ without synchronization costs. In this case, for any given schedule, the asynchronous cost is always at most as high as the synchronous cost.

One natural question here is how different the two optimal solutions can be from each other: if we find the optimum for synchronous cost, then in terms of asynchronous cost, how far away can this be from the optimum, and vice versa?

We show in two simple examples that this difference is indeed notable. We first begin with the case when we optimize for asynchronous cost, and evaluate this with respect to synchronous cost.

\begin{lemma} \label{lem:async1}
Consider an optimal solution for asynchronous cost. For synchronous cost, this can be a $\frac{P}{2}_{\!}-_{\!}\varepsilon$ factor away from the optimum, for any $\varepsilon\!>\!0$.
\end{lemma}

In the reverse case, showing such a high factor difference is more challenging. Instead, we present a simpler example that already demonstrates a $1.33$ factor.

\begin{lemma} \label{lem:async2}
Consider an optimal solution for synchronous cost. For asynchronous cost, this can be a $\frac{4}{3}_{\!}-_{\!}\varepsilon$ factor away from the optimum, for any $\varepsilon\!>\!0$.
\end{lemma}

On a high level, Lemma~\ref{lem:async1} uses a construction where we sort our $P$ processors into $\frac{P}{2}$ pairs, and each pair of processors has a very costly computation in a different one of $\frac{P}{2}$ supersteps. The optimal asynchronous solution does not care about aligning the supersteps for the different processor pairs, since they do not affect each other. However, with synchronization, the same solution incurs a high cost in each of the $\frac{P}{2}$ supersteps, and it is much better instead to place the costly computations all in the same superstep.

In Lemma~\ref{lem:async2}, we use a simpler construction with only a few nodes, where the synchronous model motivates us to place two large computations into the same superstep in order to only have one superstep of high cost; however, this results in an unnecessary increase with respect to asynchronous cost.

\section{An ILP-based solution} \label{sec:ilp}

One promising approach to solve a problem as complex as MBSP scheduling is to formulate it as an Integer Linear Program (ILP), and apply a modern ILP solver to find a solution. Naturally, the ILP representation of the problem requires a very high number of variables, and as such, this is only viable on moderate-sized DAGs. However, even on these smaller DAGs, the returned schedules can give us vital insights into the properties of the problem. In this section, we outline this ILP-based technique and discuss some fundamental optimizations on it.

\subsection{ILP representation} \label{sec:full_ILP}

The key binary variables in the natural ILP representation indicate the main aspects of a pebbling sequence. In particular, for each node $v$, processor $p$ and time step $t$, we introduce variables $\textsc{compute}_{p,v,t}$, $\textsc{save}_{p,v,t}$ and $\textsc{load}_{p,v,t}$ to indicate whether node $v$ is computed on $p$, saved to RAM from $p$, or loaded into the cache of $p$, respectively, in time step $t$. Besides this, we also add variables $\textsc{hasred}_{p,v,t}$ and $\textsc{hasblue}_{v,t}$ to indicate whether node $v$ contains a red pebble of processor $p$ or a blue pebble, respectively, in the beginning of time step $t$. The deletion operations are not directly represented as steps in this ILP; instead, they are captured implicitly when we have $\textsc{hasred}_{p,v,t}=1$, but $\textsc{hasred}_{p,v,(t+1)}=0$.

Given these variables, the fundamental properties of the schedule can be expressed naturally with linear constraints. The fundamental constraints are summarized in Figure~\ref{fig:ilp}. For instance, the prerequisites of specific operations are ensured via the constraints in lines (1)--(3). Constraints (4)--(5) define how we can obtain a new pebble on a node. Constraint (6) ensures that a processor only executes a single operation in each step. The memory bound is enforced via line (7). Finally, constraints (8)--(10) enforce the appropriate initial and terminal state.

\begin{figure*}[t]
    \centering
    \begin{tikzpicture}

    \node[anchor=center] at (-69.5pt,172pt) { $\textsc{load}_{p,v,t} \: \leq \: \textsc{hasblue}_{v,t}$};
     \node[anchor=west] at (155pt,172pt) { $\forall$ $v$, $p$, $t$};
    
    \node[anchor=center] at (-68pt,154pt) { $\textsc{save}_{p,v,t} \: \leq \: \textsc{hasred}_{p,v,t}$};
     \node[anchor=west] at (155pt,154pt) { $\forall$ $v$, $p$, $t$};

    \node[anchor=center] at (-76pt,136pt) { $\textsc{compute}_{p,v,t} \: \leq \: \textsc{hasred}_{p,u,t}$};
     \node[anchor=west] at (155pt,136pt) { $\forall$ $p$, $t$, $\forall$  $(u, v) \! \in \! E$};
    
    \node[anchor=center] at (0pt,118pt) { $\textsc{hasred}_{p,v,t} \: \leq \: \textsc{hasred}_{p,v,(t-1)} + \textsc{compute}_{p,v,(t-1)} + \textsc{load}_{p,v,(t-1)}$};
     \node[anchor=west] at (155pt,118pt) { $\forall$ $v$, $p$, $\forall$ $t \geq 1$};

     \node[anchor=center] at (-30pt,100pt) { $\textsc{hasblue}_{v,t} \: \leq \: \textsc{hasblue}_{v,(t-1)} + \sum_p \, \textsc{save}_{p,v,(t-1)}$};
     \node[anchor=west] at (155pt,100pt) { $\forall$ $v$, $p$, $\forall$ $t \geq 1$};

    \node[anchor=center] at (-2pt,82pt) { $1 \: \geq \: \sum_{v} \, \left( \textsc{compute}_{p,v,t} + \textsc{save}_{p,v,t} + \textsc{load}_{p,v,t}\right)$};
     \node[anchor=west] at (155pt,82pt) { $\forall$ $p$, $t$};

    \node[anchor=center] at (-34pt,64pt) { $ r \: \geq \: \sum_v \, \mu(v) \cdot \textsc{hasred}_{p,v,t}$};
     \node[anchor=west] at (155pt,64pt) { $\forall$ $p$, $t$};

    \node[anchor=center] at (-92pt,46pt) { $\textsc{hasred}_{p,v,0} \: = \: 0$};
     \node[anchor=west] at (155pt,46pt) { $\forall$ $v$, $p$};

     \node[anchor=center] at (-40pt,28pt) { $\textsc{hasblue}_{v,0} \: = \: 1$ if $v$ is a source, $0$ otherwise};
     \node[anchor=west] at (155pt,28pt) { $\forall$ $v$};

     \node[anchor=center] at (-97pt,10pt) { $\sum_t \, \textsc{hasblue}_{v,t} \: = \: 1$};
     \node[anchor=west] at (155pt,10pt) { $\forall$ sink node $v$};

     \node[anchor=center] at (-190pt,10pt) { $(10)$};
     \node[anchor=center] at (-190pt,28pt) { $(9)$};
     \node[anchor=center] at (-190pt,46pt) { $(8)$};
     \node[anchor=center] at (-190pt,64pt) { $(7)$};
     \node[anchor=center] at (-190pt,82pt) { $(6)$};
     \node[anchor=center] at (-190pt,100pt) { $(5)$};
     \node[anchor=center] at (-190pt,118pt) { $(4)$};
     \node[anchor=center] at (-190pt,136pt) { $(3)$};
    \node[anchor=center] at (-190pt,154pt) { $(2)$};
    \node[anchor=center] at (-190pt,172pt) { $(1)$};

\end{tikzpicture}
    \caption{Fundamental linear constraints in the ILP formulation of MBSP scheduling.}
    \label{fig:ilp}
\end{figure*}

The ingredients of the cost functions can also be expressed with further auxiliary variables and linear constraints, but this is more technical. The asynchronous cost function is the simpler case: here we need continuous variables $\textsc{finishtime}_{p,t}$ and $\textsc{getsblue}_{v}$ to fulfill the roles of the functions $\gamma$ and $\Gamma$, respectively. We need to enforce that $\textsc{finishtime}_{p,t} - \textsc{finishtime}_{p,(t-1)}$ is always larger than the cost of step $t$, which can be expressed simply as
\[ \sum_v \, \omega(v) \cdot \textsc{compute}_{p,v,t} + g \cdot \mu(v) \cdot (\textsc{save}_{p,v,t} + \textsc{load}_{p,v,t})  \, . \]
For the dependence on the slow memory, we also have
\begin{equation*}
 \textsc{getsblue}_{v} \geq \textsc{finishtime}_{p,t} - M \cdot (1 - \textsc{save}_{p,v,t}) \, ,
\end{equation*}
where $M$ is large synthetic parameter, so the constraint only has any effect when $\textsc{save}_{p,v,t} = 1$. This motivates the ILP solver to save each node only once; the finishing time of this step becomes a lower bound on $\textsc{getsblue}_{v}$. Similarly, the constraint
\begin{equation*}
\textsc{finishtime}_{p,t} \geq \textsc{getsblue}_{v} + g \cdot \mu(v) - M \cdot (1 - \textsc{load}_{p,v,t})
\end{equation*}
expresses that if $\textsc{load}_{p,v,t}=1$, then this operation can only start at $\textsc{getsblue}_{v}$, and hence it finishes at $\textsc{getsblue}_{v} + g \cdot \mu(v)$ at the earliest. Finally, a continuous variable $\textsc{makespan}$ is used with the constraints $\textsc{finishtime}_{p,t} \leq \textsc{makespan}$, and the objective function of the ILP is simply to minimize this $\textsc{makespan}$.

The synchronous model is significantly more complex, since we need to capture the separate supersteps in the schedule to express the cost function. On a high level, we first create binary variables $\textsc{compphase}_{t}$, $\textsc{savephase}_{t}$ and $\textsc{loadphase}_{t}$ to indicate the type of each time step $t$. The constraints ensure e.g.\ that $\textsc{compphase}_{t}=1$ if we have $\textsc{compute}_{p,v,t}=1$ for any $v$ and $p$, and we can also ensures with $\textsc{compphase}_{t} + \textsc{savephase}_{t} + \textsc{loadphase}_{t} \leq 1$ that the schedule is synchronous.

We then define binary variables $\textsc{compends}_{t}$ to identify the endpoints of each compute phase using $\textsc{compends}_{t} \leq \textsc{compphase}_{t}$ and $\textsc{compends}_{t} \geq \textsc{compphase}_{t} - \textsc{compphase}_{(t+1)}$. 
Then a continuous variable $\textsc{compuntil}_{p,t}$ is used to keep summing up the compute costs of subsequent steps on processors $p$, but we allow to reset this sum to $0$ right before the beginning of each compute phase (when $\textsc{commends}_{t} = 1$). Specifically, $\textsc{compuntil}_{p,t}$ is lower bounded by
\begin{equation*}
 \textsc{compuntil}_{p,(t-1)} + \sum_v \omega(v) \cdot \textsc{compute}_{p,v,t} - M \cdot \textsc{commends}_{t} \, .
\end{equation*}
Finally, we define a continuous variable $\textsc{compinduced}_{t}$ which takes the maximal value of these summed-up costs when we are in the last step of a compute phase, and can be set to $0$ otherwise:
\[
 \textsc{compinduced}_{t} \geq \textsc{compuntil}_{p,t} - M \cdot (1 - \textsc{compends}_{t}) \, . \]
The I/O costs can be expressed in a very similar way. The overall objective of the ILP problem is then to minimize
\[ \sum_t \, \textsc{compinduced}_{t} + \textsc{comminduced}_{t} + L \cdot \textsc{commends}_{t} \, . \]

For more details on the ILP representation, we refer the reader to the full version of the paper, or the source codes of the algorithms.

\subsection{The number of time steps} \label{sec:time}

Note that one crucial degree of freedom in this ILP representation is the upper bound $T$ on the number of time steps allowed. While a too small $T$ may exclude the optimal solution, an unnecessarily high $T$ creates too many variables in our ILP. In order to significantly reduce the number of time steps, we apply \emph{step merging}: we allow to combine multiple steps of our MBSP schedule into a single ILP time step. Indeed, for I/O operations, we can allow to save and/or load multiple values in the same step if all their prerequisites are satisfied simultaneously beforehand. Similarly, we can merge consecutive compute operations on a processor into a single step if all their inputs and outputs fit into cache simultaneously, even when there is a precedence relation between the nodes. Naturally, this step merging also requires to significantly adapt many of the constraints in our ILP accordingly.

The choice of $T$ also raises an interesting side question: if we happen to select $T$ too small and accidentally exclude the optimal solution, is there a simple way to notice this? For instance, if our ILP schedule contains an empty step with no operation, then one might assume that the solution is indeed optimal, since the solver decided not to use this step. We show that somewhat counter-intuitively, this is not the case.

\begin{lemma} \label{lem:ilp}
Assume that the optimal schedule of the ILP restricted to $T_0$ steps has cost $C_0$ and contains an empty step. It can happen that for some $T > T_0$, the optimal ILP schedule restricted to $T$ steps has a smaller cost $C < C_0$.
\end{lemma}

This lemma already holds for $P=1$ and uniform node weights. On a high level, the proof uses a construction where some I/O steps in our schedule can be substituted by re-computing a chain of $d$ nodes instead. This reduces the total cost by $g-d$, but requires $(d-1)$ further steps (which also cannot be merged). Altogether, this means that if $g \geq d$ and if we select $T$ appropriately, then it can happen that the ILP schedule contains up to $(d-2)$ empty steps, but it is still suboptimal.

\subsection{Divide-and-Conquer method} \label{sec:dnq}

The number of variables in the above ILP representation scales with the number of nodes, processors and time steps, and thus even with very powerful modern ILP solvers, it already becomes untractable for DAGs with a few hundred nodes in practice. We also present a divide-and conquer method to extend the ILP-based approach to slightly larger DAGs. This consists of the following steps:

\begin{enumerate}[topsep=3pt, leftmargin=1.6em, itemsep=4pt, label=\textbf{\arabic*)}]
 \item First, our goal is to partition the input DAG into smaller parts such that the corresponding quotient graph remains acyclic. This allows us to develop a schedule for each part independently according to a topological order of the parts. We also want to have as few edges as possible between the separate parts; intuitively, this helps to ensure that the subproblems can indeed be solved independently without a major negative effect on the combined schedule. This acyclic partitioning problem can also be expressed as an ILP, like many graph partitioning problems before~\cite{jenneskens22}. Moreover, the size of this ILP problem is drastically smaller than that of a scheduling problem on the same DAG, since we now do not have a time dimension in our variables. Indeed, our experiments also showed that for acyclic partitioning problems with only two parts, the ILP solver almost always found the optimum solution in negligible time. Hence as our first step, we use this ILP-based acyclic bipartitioning method to recursively split the DAG into smaller subDAGs, until each of our subDAGs consist of at most $60$ nodes.
 \item We then develop a high-level `scheduling plan' for the partitioned DAG by considering the quotient graph, and using an adjusted version of the BSPg scheduling heuristic~\cite{BSP_algos_opdas} that allows to assign multiple processors to a specific node to reduce its computation time proportionally. This gives us a high-level schedule on the quotient DAG, defining which set of processors will be used for each subproblem. For DAGs that are close to sequential, this may simply mean that we consider the parts in topological order, and use all available processors in each subproblem. However, when the subDAGs are more parallel to each other, the available processors may be split between the given subproblems.
 \item We then consider the subproblems (i.e., subDAGs and subsets of processors) in a topological order, and use our ILP-based scheduler to find a good schedule for each of these subproblems. Note that when solving these MBSP subproblems, we also require further modifications to the ILP representation above: for instance, some nodes might already have a red pebble in the beginning (leftover from the previous subschedule), and non-sink nodes may also require a blue pebble by the end of the schedule (if they have children in a following subDAG). Note that in each subproblem, we always allow to use the entire cache capacity of the given processors.
 \item Finally, the subDAG schedules are concatenated into an MBSP schedule for the whole problem. A few more technical steps are executed to streamline this combined schedule, and to remove some unnecessary steps that were created due to the split.
\end{enumerate}

Note that this divide-and-conquer ILP is more of a heuristic approach to the problem: even if all subILPs are solved to optimality, this does not ensure that the combined MBSP schedule is a global optimum. Indeed, the ILPs now only capture specific parts of the problem, and they may ignore e.g.\ that some values would be vital to keep in cache for the following subDAG. Our experiments also confirm that on some DAGs, this method can return a worse MBSP schedule than the baseline. Nonetheless, when we manage to split the DAG into relatively disjoint parts, the divide-and-conquer ILP often allows to find notably better schedules on larger instances.

\section{Experiments and results} \label{sec:experiments}

\subsection{Experimental setup}

For our experiments, we compare the above ILP-based schedulers to a two-stage approach. For the scheduling stage, we use the recent BSPg scheduling heuristic designed for the BSP model~\cite{BSP_algos_opdas}, which focuses both on workload balancing and minimizing communication. For the memory management stage, we use the clairvoyant algorithm that considers all the following compute steps on the same processor, and whenever having to evict a value from cache, it selects the value that is not required for the longest time. Both of these algorithms excel at their respective task. We will also briefly consider two-stage baselines based on some other algorithms: the Cilk work-stealing heuristic~\cite{cilk}, an ILP-based BSP-scheduler, and the `least recently used' (LRU) cache eviction rule.

To compare the schedulers, we use the computational DAG benchmark introduced in~\cite{BSP_algos_opdas}. The smallest dataset here consists of $15$ DAGs between $40$ and $80$ nodes, describing a variety of linear algebra and graph computations. In particular, $3$ of the DAGs are coarse-grained representations of specific algorithms (BiCGSTAB, $k$-means, Pregel), and the others are more fine-grained instances of four specific problems (CG, SpMV, iterated SpMV and $k$-NN). The dataset has compute weights $\omega$, but not memory weights $\mu$, hence we add uniform random memory weights from $\{1, ..., 5\}$ to each node of the DAGs.

Since the efficiency of the ILP method can be significantly improved by starting from a good initial solution, we initialize the solvers with our baseline, and study how much our ILP schedulers can further improve on this solution. The size of the DAGs allows us to use the full ILP formulation from Section~\ref{sec:full_ILP}.

We also consider a sample from the second smallest dataset in~\cite{BSP_algos_opdas}, taking the two smallest instances for each DAG type (coarse-grained, CG, SpMV, iterated SpMV, $k$-NN). This gives us $10$ DAGs between $264$ and $464$ nodes, where applying the full ILP formulation is not viable anymore; we use this dataset to study the divide-and-conquer ILP.

Due to the relatively small DAGs, we use $P\!=\!4$ processors for most experiments. For each DAG in our experiments, we define the select size $r$ with respect to the minimal memory $r_0$ required to allow a valid schedule on the DAG (i.e.\ $r_0$ the maximal sum of memory weights for a node and its parents). In particular, in our main experiments, we set $r\! =\! 3\! \cdot\! r_0$ for each DAG. For the rest of the BSP parameters, we pick $g\!=\!1$ and $L\!=\!10$, similarly to~\cite{BSP_algos_opdas}.

For solving the actual ILP problems, we use the COPT (Cardinal Optimizer) commercial ILP solver~\cite{copt}, version 7.1.7. This is a state-of-the-art commercial solver that uses several threads in parallel to optimize the problem. For the ILP representations of the entire scheduling problem, we run the solver with a time limit of $60$ minutes; in case of the divide-and-conquer approach, we apply a time limit of $30$ minutes to each subproblem. This is typically not enough to run the solving process to optimality, but already allows to find very good solutions.

We run our experiments on an AMD EPYC 7763 processor, using 64 cores and 156 GB of memory. The implementations of our algorithms, the baselines and the test suite are available open-source as part of our OneStopParellel scheduling framework on GitHub:
\url{https://github.com/Algebraic-Programming/OneStopParallel}.

\subsection{Results}

In general, our results show that the holistic ILP-based method can often find notably better schedules than the two-stage baseline. In particular, in our main experiments with the default parameter values, the ILPs obtain a $0.77\times$ factor geometric-mean reduction in the synchronous costs compared to the baselines, which ranges from $0.99\times$ to $0.6\times$ on the concrete instances. The cost of the baseline and the ILP schedules for each instance is listed for completeness in Table~\ref{tab:tiny_dataset}. The reduction factors for this base case (and several of the following cases) are also illustrated in Figure~\ref{fig:plot}.

\begin{figure*}[t]
    \centering
    \hspace{-12pt}
    \includegraphics[scale=0.67]{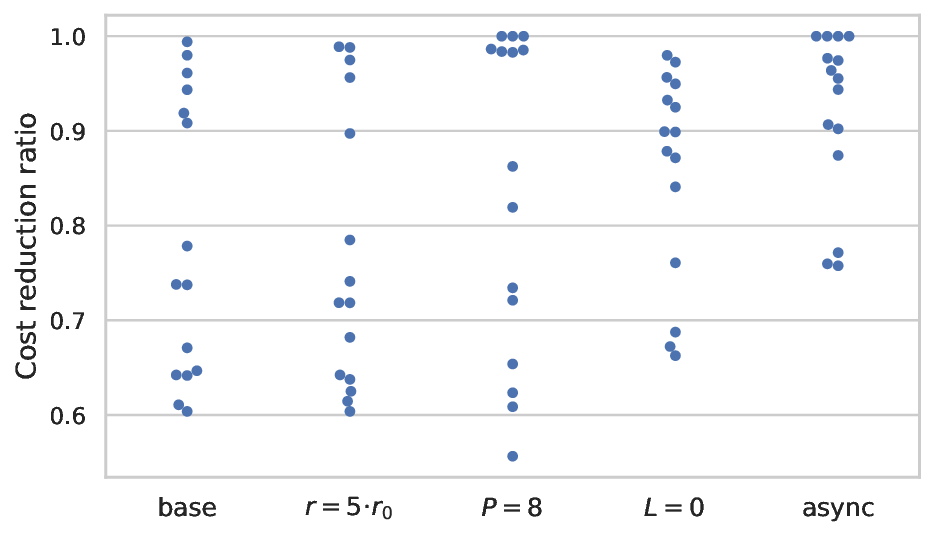}
    \caption{Illustration of the distribution of cost reduction factors for the smallest dataset of~\cite{BSP_algos_opdas}, for several different choices of parameters. The base case has $P=4$, $r = 3 \!\cdot\! r_0$ and $L=10$.}
    \label{fig:plot}
\end{figure*}

\begin{table}[t]

        \renewcommand{\arraystretch}{1.3}
        \caption{Concrete cost of synchronous MBSP schedules with the two-stage baseline / our ILP method, on each computational DAG of the smallest dataset in~\cite{BSP_algos_opdas}.}\label{tab:tiny_dataset}
        \hspace{-0.03\textwidth}
        \begin{minipage}{0.21\textwidth}
            \centering
            \begin{tabular}{ c | c |}
        		$\:$ Instance $\:$ & $\!\!$ Base / ILP $\!\!$ \\ \hline \hline
        		$\!\!\!$ bicgstab $\!\!\!$ & $\!$ 197 / 181 $\!$ \\ \hline
                $\!\!\!$k-means $\!\!\!$ & $\!$ 158 / 106 $\!$ \\ \hline
                pregel & $\!$ 206 / 152 $\!$ \\ \hline
                spmv\_N6 & $\!$ $\!$ 123 / 79 $\!$ \\ \hline
                spmv\_N7 & $\!$ 120 / 77 $\!$ \\ \hline
                spmv\_N10 & $\!$ 159 / 96 $\!$ \\ \hline
                CG\_N2\_K2 & $\!$ 283 / 267 $\!$ \\ \hline
                CG\_N3\_K1 & $\!$ 199 / 195 $\!$ \\ \hline
	   \end{tabular}
        \end{minipage}
        \hspace{0.05\textwidth}
        \begin{minipage}{0.2\textwidth}
            \centering
            \begin{tabular}{ c | c |}
        		$\:$ Instance $\:$ & $\!\!$ Base / ILP $\!\!$ \\ \hline \hline
                $\!\!$ CG\_N4\_K1 $\!\!$ & $\!$ 229 / 208 $\!$ \\ \hline
                $\!\!$ exp\_N4\_K2 $\!\!$ & $\!$ 149 / 91 $\!$ \\ \hline
                $\!\!$ exp\_N5\_K3 $\!\!$ & $\!$ 185 / 144 $\!$ \\ \hline
                $\!\!$ exp\_N6\_K4 $\!\!$ & $\!$ 169 / 168 $\!$ \\ \hline
                $\!\!$ kNN\_N4\_K3 $\!\!$  & $\!$ 179 / 132 $\!$ \\ \hline
                $\!\!$ kNN\_N5\_K3 $\!\!$ & $\!$ 167 / 108 $\!$ \\ \hline
                $\!\!$ kNN\_N6\_K4 $\!\!$ & $\!$ 180 / 173 $\!$ \\ \hline
	   \end{tabular}
            \vspace{14pt}
        \end{minipage}
        \hspace{-0.02\textwidth}
\end{table}

In order to gain a deeper understanding of the problem, we analyze how these improvement factors depend on different parameters of the model. For instance, if we use a larger memory bound of $r \! = \! 5 \! \cdot \! r_0$, we obtain very similar results: the cost here is reduced by a $0.76\times$ factor (geo-mean). However, in e.g.\ the edge case when we select the minimal $r \! = \! r_0$, the cost reduction is drastically smaller ($0.97\times$), with no improvement on $6$ out of $15$ instances. This is likely because an MBSP schedule with this $r$ is very restricted, leaving notably less degree of freedom for the ILP to improve upon the baseline.

If we increase the number of processors to $P\!=\!8$, we observe a slightly worse improvement ratio of $0.82\times$: the ILPs are more challenging to optimize in this case, since they contain almost twice as many variables.

If we still consider synchronous scheduling model, but without synchronization cost ($L=0$), then the cost reduction slightly drops, to a $0.85\times$ factor. If we switch to the asynchronous cost model instead, our methods only obtain a $0.91\times$ geomean cost reduction to the baselines. In general, our experience is that this asynchronous ILP formulation is more challenging to optimize for the ILP solver, likely due to the intricate dependencies between the non-binary $\textsc{finishtime}_{p,t}$ variables that express $\gamma(\tau_i)$.

It is also a natural idea to compare our ILP schedulers to other two-stage baselines. As a stronger baseline, we can already formulate and optimize the BSP scheduling problem as an separate ILP in the first stage, similarly to~\cite{BSP_algos_opdas}. We found that our method can still obtain a geo-mean improvement of $0.88\times$ over these more sophisticated initial schedules. Since the first stage uses the same ILP solver in this case, this again confirms the superiority of the holistic method over the two-stage approach.

It is also interesting to consider a more application-oriented two-stage baseline with well-known methods from practice. For this, we first use the Cilk work-stealing scheduler to find an efficient schedule, and then the `least recently used' (LRU) cache eviction rule for memory management. These generic heuristics may be more representative baselines in many applications than the specialized heuristics above. When compared to this baseline, our MBSP schedules amount to an even larger, $0.66\times$ geomean cost reduction.

While we focus on multiprocessor scheduling, the ILPs also give us a unique chance to analyze the single-processor case with $P\!=\!1$. This is essentially the red-blue pebble game of Hong and Kung~\cite{hongkung} with computation costs and node weights. The red-blue pebble game has been studied extensively for decades, but surprisingly, we are not aware of any experimental study of pebbling algorithms. We can simply use a DFS as our first-stage scheduler here, and we already obtain a very strong baseline for red-blue pebbling: we found that the ILPs could only improve on this baseline in $2$ out of the $15$ instances. This suggests that the strength of our scheduling approach is indeed in the fact that it can consider the multiprocessor scheduling and the memory management problems holistically.

Another interesting option is to prohibit re-computation in our schedules, i.e.\ to require that each node is only computed on one processor and only once. This simplifying assumption is used in many scheduling and pebbling models. In general, re-computations can reduce the total cost by allowing to save I/O steps; their effect has been studied theoretically before~\cite{recent4}, but not experimentally, to our knowledge. Our framework offers a convenient opportunity to study this question, since we can easily prohibit re-computation in our ILPs with some further constraints. If we do this, our experiments show up to a $1.4\times$ cost increase on some specific instances (compared to the default ILP that allows recomputation). This indicates that re-computation steps are indeed actively used in efficient schedules.

Finally, we also briefly consider how the ILP-based scheduling approach scales to larger DAGs. We consider the larger dataset of $10$ instances with $264$--$464$ nodes, using $r = 5 \! \cdot \! r_0$. Recall that these DAGs generate much larger ILP formulations, and hence we apply the divide-and-conquer algorithm on this dataset. While this method excels at optimizing each subproblem, it does not optimize for a global optimum, and hence it may fall behind the baseline. Indeed, we found that our method finds notably better schedules on specific kinds of graphs, e.g.\ the coarse-grained or the SpMV computations (see~\cite{BSP_algos_opdas}), achieving a $0.66\times$ and $0.75\times$ geomean cost reduction, respectively. However, on the remaining DAGs which do not allow for so good partitionings into loosely connected parts, the schedules returned are actually a $1.13\times$ geomean factor worse than the baseline. The concrete costs for each instance are listed in Table~\ref{tab:small_dataset}. We find it a promising direction for future work to better understand this effect, and to identify specific applications or families of DAGs where our divide-and-conquer approach can reliably outperform the baseline.

\begin{table}[!htp]

        \renewcommand{\arraystretch}{1.3}
        \caption{Cost of MBSP schedules with the baseline / our divide-and-conquer ILP, on some larger computational DAGs from~\cite{BSP_algos_opdas}. For the instances on the left side, our method can indeed find a better schedule; for those on the right, it could not outperform the two-stage baseline.}\label{tab:small_dataset}
        \hspace{-0.04\textwidth}
        \begin{minipage}{0.21\textwidth}
            \centering
            \begin{tabular}{ c | c |}
        		$\:$ Instance $\:$ & $\!\!\!$ Base / ILP $\!\!\!$ \\ \hline \hline
        		$\!\!\!$ simple\_pagerank $\!\!\!$ & $\!$ 1017 / 779 $\!$ \\ \hline
                $\!\!\!$snni\_graphchall. $\!\!\!$ & $\!$ 1531 / 912 $\!$ \\ \hline
                spmv\_N25 & $\!$ 425 / 314 $\!$ \\ \hline
                spmv\_N35 & $\!$ $\!$ 685 / 518 $\!$ \\ \hline
                CG\_N5\_K4 & $\!$ 847 / 750 $\!$ \\ \hline
	   \end{tabular}
        \end{minipage}
        \hspace{0.05\textwidth}
        \begin{minipage}{0.2\textwidth}
            \centering
            \begin{tabular}{ c | c |}
        		$\:$ Instance $\:$ & $\!\!\!$ Base / ILP $\!\!\!$ \\ \hline \hline
                $\!\!$ CG\_N7\_K2 $\!\!$ & $\!$ 701 / 701 $\!$ \\ \hline
                $\!\!$ exp\_N10\_K8 $\!\!$ & $\!$ 573 / 727 $\!$ \\ \hline
                $\!\!$ exp\_N15\_K4 $\!\!$ & $\!$ 512 / 660 $\!$ \\ \hline
                $\!\!$ kNN\_N10\_K8 $\!\!$  & $\!$ 594 / 682 $\!$ \\ \hline
                $\!\!$ kNN\_N15\_K4 $\!\!$ & $\!$ 517 / 655 $\!$ \\ \hline
	   \end{tabular}
        \end{minipage}
        \hspace{-0.03\textwidth}
\end{table}

\section{Conclusion}

Altogether, our work indicates that multiprocessor scheduling with memory constraints is a rather delicate problem on general DAGs. We have shown that a two-stage approach can be highly suboptimal in theory, and that an ILP-based holistic method can indeed return notably better schedules in many cases.

Naturally, one significant drawback of this ILP method is that it takes significantly more time to find a schedule than lightweight scheduling heuristics. As such, it is not directly applicable in many domains where the schedule computation time is also a critical factor. However, there are also many cases where the same static computation is executed many times (possibly with different inputs), e.g.\ for an important low-level operator in a neural network. In these settings, even a small improvement in the scheduling cost can have a significant impact on the long term. Furthermore, even in domains where our ILP-based method is not directly applicable, the algorithm can give valuable insights into how to design efficient schedules, or how far a scheduling heuristic is from the optimal solution.

\bibliographystyle{ACM-Reference-Format}
\bibliography{references}

\clearpage

\appendix

\section{Model definition details} \label{app:model}

We begin with some technical details omitted from the MBSP problem definition.

Throughout the pebbling, we formally define the current state of a processor $p \in [P]$ as a tuple $(R_p, B_p)$, where $R_p$ is the set of nodes that currently have a red pebble of processor $p$ on them, and $B_p$ is the set of nodes that have a blue pebble according to the current knowledge of processor $p$ (also considering the last synchronization point). Note that this $B_p$ is a rather artificial concept, since the processors actually share a common slow memory, but it allows to define the pebbling sequences on the specific processors independently from each other.

For any node $v \!\in \! V$, let $\text{Par}(v)$ and $\text{Chld}(v)$ denote the set of parents and children of $v$, respectively. A state $(R_p, B_p)$ is initial if only the source nodes of $G$ have blue pebbles, i.e., $B_p(0)=\{ v \in G \, | \, \text{Par}(v) = \emptyset \}$, and $R_p(0) = \emptyset$. A state $(R_p, B_p)$ is terminal if all the sink nodes have a blue pebble, i.e., $\{ v \in G \, | \, \text{Chld}(v) = \emptyset \} \, \subseteq \, B_p$. Given a current state of the pebbling $(R_p(t), B_p(t))$ on processor $p$, for a specific node $v \in V$, the transition rules are as follows:
\begin{enumerate}[topsep=4pt, itemsep=3pt]
\item LOAD$_{p,v}$: place a red pebble of processor $p$ on node $v$ if $v$ already has a blue pebble. I.e., if $v \in B_p(t)$, then set $R_p(t+1)=R_p(t) \cup \{v\}$ and $B_p(t+1)=B_p(t)$.
\item SAVE$_{p,v}$: place a blue pebble on node $v$ if $v$ already has a red pebble of processor $p$. I.e., if $v \in R_p(t)$, then set $R_p(t+1)=R_p(t)$ and $B_p(t+1)=B_p(t) \cup \{v\}$.
\item COMPUTE$_{p,v}$: if all the parents of a non-source node $v$ have a red pebble of processor $p$, then place a red pebble of processor $p$ on $v$. I.e, if $\text{Par}(v) \subseteq R_p(t)$ with $\text{Par}(v) \neq \emptyset$, then set $R_p(t+1)=R_p(t) \cup \{v\}$ and $B_p(t+1)=B_p(t)$.
\item DELETE$_{p,v}$: remove the red pebble of processor $p$ from node $v$. I.e., if $v \in R_p(t)$, then set $R_p(t+1)=R_p(t) \setminus \{v\}$ and $B_p(t+1)=B_p(t)$.
\end{enumerate}
A pebbling sequence on processor $p$ is a sequence of states
\[ ((R_p(0), B_p(0)), \: (R_p(1), B_p(1)), \: ..., \: (R_p(T), B_p(T))) \, , \]
or a sequence of transition rules $\Psi_p = (\tau_1, ..., \tau_T)$ from the list above, where for all $t \in [T]$, the prerequisites of transition rule $\tau_i$ are fulfilled in $(R_p(t-1), B_p(t-1))$, and we obtain $(R_p(t), B_p(t))$ from $(R_p(t-1), B_p(t-1))$ by applying rule $\tau_i$. We also require that the sequence is valid, i.e.\ it always fulfills the memory bound: for all $t \in [T] \cup \{ 0 \}$, we have $\sum_{v \in R_p(t)} \, \mu(v) \leq r$ . The cost of the sequence is
\[ cost(\Psi_p) = \sum_{t \in [T]} \, cost(\tau_t) \, . \]

Let the starting/final state of the red/blue pebbles be denoted by $R_{start}(\Psi_p)=R_p(0)$, $B_{start}(\Psi_p)=B_p(0)$, $R_{end}(\Psi_p)=R_p(T)$, $B_{end}(\Psi_p)=B_p(T)$.

Recall from before that a superstep $\Psi_p$ on processor $p$ is a $4$-tuple $(\Psi_{p, comp}, \Psi_{p, save}, \Psi_{p, del}, \Psi_{p, load})$. Previously we left out the dependence on $p$ from the notation for simplicity.

A superstep $\Psi$ is then a tuple $\Psi=(\Psi_1, ..., \Psi_P)$, where $\Psi_p$ is a superstep on processor $p$ for $p \in [P]$, and for all $p \in [P]$, we have
\begin{itemize}[topsep=4pt, itemsep=2pt]
 \item $R_{start}(\Psi_{p, save}) = R_{end}(\Psi_{p, comp})$, 
 \item $R_{start}(\Psi_{p, del}) = R_{end}(\Psi_{p, save})$,
 \item $R_{start}(\Psi_{p, load})  = R_{end}(\Psi_{p, del})$,
 \item $B_{start}(\Psi_{p, save}) = B_{end}(\Psi_{p, comp})$,
 \item $ B_{start}(\Psi_{p, del}) = \bigcup_{p \in [P]} \, B_{end}(\Psi_{p, save})$,
 \item $B_{start}(\Psi_{p, load}) = B_{end}(\Psi_{p, del})$.
\end{itemize}
An MBSP schedule is a sequence of supersteps $(\Psi^{(1)}, \Psi^{(2)}, ..., \Psi^{(m)})$ such that:
\begin{itemize}
 \item for the superstep $\Psi^{(1)}=(\Psi^{(1)}_1, ..., \Psi^{(1)}_P)$, for any $p \in [P]$, the state
 \[ \left( \, R_{start}\left(\Psi^{(1)}_{p, comp}\right), \: B_{start}\left(\Psi^{(1)}_{p, comp}\right) \, \right) \]
 is an initial state;
 \item for any two consecutive supersteps $\Psi=(\Psi_1, ..., \Psi_P)$ and $\Psi'=(\Psi'_1, ..., \Psi'_P)$ in the sequence, for any $p \in [P]$, we have
 \begin{gather*}
 R_{start}(\Psi'_{p, comp}) = R_{end}(\Psi_{p, load}) \, , \\
 B_{start}(\Psi'_{p, comp}) = B_{end}(\Psi_{p, load}) \, ;
 \end{gather*}
 \item for the superstep $\Psi^{(m)}=(\Psi^{(m)}_1, ..., \Psi^{(m)}_P)$, for any $p \in [P]$, the state
 \[ \left( \, R_{end}\left(\Psi^{(m)}_{p, load}\right), \: B_{end}\left(\Psi^{(m)}_{p, load}\right) \, \right) \]
 is a terminal state.
 \end{itemize}

In the definition of the synchronous cost function, we again omitted the dependence on processors for simplicity: for instance, $cost(\Psi_{comp})$ should instead be $cost(\Psi_{p, comp})$.

We also note that in this synchronous cost function, one might wonder if it is more natural to allow a processor to load values while another one is saving, and hence define $cost(\Psi)$ as
\begin{equation*}
\max_{p \in [P]} \, cost(\Psi_{p,comp}) + \max_{p \in [P]} \, \left( cost(\Psi_{p,save}) + cost(\Psi_{p,load}) \right) + L \, 
\end{equation*}
instead. However, such a model variant would require a more complex definition for communication phases in order to ensure consistency, i.e.\ that values are only loaded from slow memory after they become available there. This model variant would also not be a direct generalization of multiprocessor red-blue pebbling.

In the asynchronous cost function, we consider the entire sequence
\begin{gather*}
\Upsilon_p := \Psi^{(1)}_{p, comp} \circ \Psi^{(1)}_{p, save} \circ \Psi^{(1)}_{p, del} \circ \Psi^{(1)}_{p, load} \circ \ldots \\ \ldots \circ \Psi^{(m)}_{p, comp} \circ \Psi^{(m)}_{p, save} \circ \Psi^{(m)}_{p, del} \circ \Psi^{(m)}_{p, load} \, ,
\end{gather*}
and define the finishing time function $\gamma$ on this sequence. Regarding the auxiliary function $\Gamma$ which shows when a node first becomes available in slow memory, for any node $v \in V$, let $j$ be the smallest superstep index such that $\exists \, p \in [P]$ so that there is a transition SAVE$_{p,v}$ in $\Psi_{p, save}^{(j)}$. Let us denote by $UP(v)$ the set of all SAVE$_{p,v}$ transitions $\tau_i$ in $\Psi_{p, save}^{(j)}$ for some $p \in [P]$, and then we define
\[ \Gamma(v) = \min_{\tau \in UP(v)} \, \gamma(\tau) \, . \]
Note that if the schedule is valid, then $\gamma(\tau_i)$ is indeed well-defined, since $v$ must already be saved to slow memory before (or in) the superstep where it is loaded.

\section{Proof details} \label{app:proofs}

This section provides the proofs and proof details omitted from the main part of the paper.

\subsection{Proof details for Theorem~\ref{th:twostage}} \label{app:subopt}

The main idea for the proof of Theorem~\ref{th:twostage} was already outlined in the main part of the paper. We discuss some technical details here.

Our construction consists of two groups $H_1$, $H_2$ of $d$ source nodes, and two chains $v_1, ..., v_m$ and $u_1, ..., u_m$, such that $(v_i, v_{i+1}) \in E$ and $(u_i, u_{i+1}) \in E$ for all $i \in [m-1]$. If $i \in [m]$ is an odd integer, then there is also an edge to $u_i$ from all nodes of $H_1$, and to $v_i$ from all nodes of $H_2$. If $i \in [m]$ is an even integer, then there is also an edge to $u_i$ from all nodes of $H_2$, and to $v_i$ from all nodes of $H_1$. We will also use $\text{Ch}(H_1)$ and $\text{Ch}(H_2)$ to denote the children of $H_1$ and $H_2$, respectively. Assume that the node weights are uniform, i.e.\ $\omega(v)=1$ and $\mu(v)=1$ for all $v \! \in \! V$. Assume that the size of the available cache is $r=d+2$, we have $P=2$ processors, and $L=0$ in the synchronous case. For our parameters, let $g=O(1)$, and let us set $m > 2 \! \cdot \! d \! \cdot \! g + d$ to cover several model variants. With $2\cdot(d+m)$ nodes in the DAG, we can still ensure $d=\Theta(n)$ and $m=\Theta(n)$.

Note that there is a modeling question here on how we interpret our MBSP problem in a BSP scheduling setting: in BSP (and other) DAG scheduling problems, the sources of the DAG are also regular nodes that need to be computed, whereas in MBSP, the source nodes are not computed, but directly loaded from slow memory (as in red-blue pebbling). However, this only plays a minor role in our example construction. If the nodes in $H_1$, $H_2$ need to be computed, this adds a compute cost between $d$ and $2d$ to any solution. Afterwards, exchanging the values between the two processors has a cost of $d \cdot g$ in a direct communication model with $h$-relations (as in e.g.\ \cite{BSP_DAG_opdas}). Alternatively, in a BSP variant where source nodes are still uncomputed and loaded from slow memory, it takes a cost of $2\cdot d \cdot g$ to load all the source nodes to both processors. As in BSP, we assume that the sinks are simply computed but not saved or sent anywhere.

In any of these BSP variants, the optimal scheduling strategy is to compute $v_1, ..., v_m$ on one processor, and $u_1, ..., u_m$ on the other. This has a compute cost of only $m$ or $m+d$, which is the minimal possible compute cost in the respective BSP variant. The communication cost is only $d \cdot g$ or $2\cdot d \cdot g$, as discussed above. We show that the BSP cost of any other schedule is larger than this reference solution. Indeed, in both of the BSP variants, this is already the minimal amount of communication cost that can be incurred if both processors acquire all the source nodes, i.e.\ if both processors compute at least one node from both $\text{Ch}(H_1)$ and $\text{Ch}(H_2)$.

As such, there are only a few possible BSP schedules with potentially lower cost. If $\text{Ch}(H_1)$ is computed entirely on one processor, and $\text{Ch}(H_2)$ on the other, then every chain node needs to be sent between the processors, giving a communication cost of at least $(m-1) \cdot g$; with $d<m-1$, this is larger than the communication cost before. If both $\text{Ch}(H_1)$ and $\text{Ch}(H_2)$ are computed entirely on the same processor, then the compute cost is at least $2m$; with $m > 2 \! \cdot \! d \! \cdot \! g  + d$, this is also suboptimal in both BSP variants.

Finally, assume w.l.o.g.\ that $\text{Ch}(H_1)$ is computed entirely on processor $p_1$, but $\text{Ch}(H_2)$ is split between the two processors. We consider the two BSP variants separately. In the variant with loaded source nodes, loading all sources to $p_1$ already incurs a communication cost of $2\cdot d \cdot g$. In the variant with computed source nodes, if there are $x \! \geq \! 1$ nodes of $\text{Ch}(H_2)$ computed on $p_2$, then all of these nodes that are not sinks (so at least $(x-1)$ nodes) need to be sent to $p_1$ for computing their child. Besides this, making $H_2$ available on both processors also requires $d$ values to be sent, regardless of how the computation of $H_2$ is split between $p_1$ and $p_2$. This altogether gives a communication cost of at least $\frac{d+x-1}{2} \cdot g$ in the BSP model with $h$-relations. If we instead compute the entire DAG on $p_1$, then these communication steps are all saved, and compute costs are increased by at most $2d+x$. For $g=5$ and $d$ large enough, we have $\frac{d+x-1}{2} \cdot g > 2d+x$, so this indeed improves the solution. However, the compute cost is now still larger by $m+d$ than in our reference solution; with $m > 2 \cdot d \cdot g$, this is suboptimal. This establishes that the reference solution is indeed the optimal one for BSP.

Let us begin from this optimal BSP schedule and use $r=d+2$; this means that besides the current two nodes in the chain, the cache can only contain either $H_1$ or $H_2$ at any point. As such, we repeatedly need to load either $H_1$ or $H_2$ for each next chain node, and hence the number of required I/O operations becomes $d \cdot m$ on the chains.

Instead, the optimal MBSP schedule is to compute all children of $H_1$ on one processor and all children of $H_2$ on the other. This results in only $2$ I/O operations for each chain node. As such, the two-stage approach results in a compute cost of $m+O(d)$, and a communication cost of $(d \cdot m + O(d)) \cdot g$. The optimal MBSP schedule has a compute cost of $m+O(d)$, and a communication cost of $(2 \cdot m + O(d)) \cdot g$. With $d<m$ and $d=\Theta(n)$, the ratio is linear in $n$.

We note that if desired, the above proof can also be generalized to more processors as long as $P \in O(1)$. For any $P \geq 3$, we simply add $(P-2)$ copies of the following component to our DAG: a group of $d$ new source nodes, and a chain of length $m$, with each source node having an edge to each chain node. We select $m > P \! \cdot \! d \! \cdot \! g + d$. We can then use almost the same arguments as above: the optimal BSP and MBSP schedules will both use $2$ processors for the original part of the DAG as before, and a separate new processor for each of the $(P-2)$ extra components. This results in the same cost for both the two-stage approach and the MBSP optimum as before.

It is easy to see that the additional components do not affect the optimal MBSP cost; the more technical part is to check that the optimal schedule for BSP will also be the same as before. In the BSP variant where source nodes are loaded, this is simpler. If there is a processor that obtains all of $H_{1\!} \cup _{\!} H_2$, then the communication cost is already $2 \cdot d \cdot g$. If not, then every non-sink node in the original chains will be sent to some other processor, resulting in $2 \cdot (m-1)$ values sent, and a communication cost of at least $\frac{2 \cdot (m-1)}{P} \cdot g$. For $m > P \cdot d + 1$, this is already larger than $2 \cdot d \cdot g$.

On the other hand, consider the variant where source nodes are also computed. If there are at least two processors that obtain all of the values $H_{1\!} \cup _{\!} H_2$, then this incurs a communication cost of at least $d \cdot g$, since one of them must receive at least $d$ values from $H_{1\!} \cup _{\!} H_2$. If no processor obtains all of $H_{1\!} \cup _{\!} H_2$, we again get $2 \cdot (m-1)$ communicated values, and a cost of at least $\frac{2 \cdot (m-1)}{P}$. Finally, assume there is exactly one processor $p_1$ that obtains $H_{1\!} \cup _{\!} H_2$. If $p_1$ computes both original chains, we again have a compute cost of at least $2 \! \cdot \! m$. If there are $x \! \geq \! 1$ nodes in the original two chains not computed by $p$, then for each such node (apart from the sinks), its child is on a different processor, resulting in at least $(x-2)$ sent values. Besides this, at least $d$ values from $H_{1\!} \cup _{\!} H_2$ must be sent to another processor than where it was computed, summing up to a communication cost of at least $\frac{d+x-2}{P} \cdot g$. Instead computing all of the original DAG on $p$ increases the compute cost by at most $2d+x$, which is smaller for e.g.\ for $g= 3 \! \cdot \! P$ and $d$ large enough, which is still a constant choice of $g$. However, this solution is still suboptimal since its compute cost is larger than $2 \cdot m$.

\subsection{Proofs of Lemmas~\ref{lem:weakNP} and \ref{lem:strongNP}} \label{app:NPhard}

We continue with the proofs of NP-hardness for the memory management subproblem. Recall that we consider a setting here where the compute steps in all compute phases $\Psi_{comp}$ in our MBSP schedule are already fixed for every processor and superstep, but we are free to add any delete operations in the compute phases, and any save, delete and load operations in the corresponding parts of the communication phases. We also implicitly assume that the compute steps are provided such that there exists at least one valid schedule.

We assume that our goal is to minimize the I/O cost, i.e.\ the total cost of the $\Psi_{save}$ and $\Psi_{load}$ phases is the synchronous setting with $L=0$. For convenience, let us first consider a simpler case where we only need to minimize the total cost over the $\Psi_{load}$ phases.

\renewcommand{\proofname}{Proof of Lemma~\ref{lem:weakNP}}
\begin{proof}
We provide a reduction from the number partitioning problem: given integers $A=\{a_1, ..., a_m\}$ with sum $\alpha$, we need to find a subset $A_0 \subseteq A$ that sums up to $\frac{\alpha}{2}$. We can reduce this setting to a memory management problem with $r= \alpha$, and separate source nodes $v_1, ..., v_m, v'$ with memory weights $a_1, ..., a_m, \frac{\alpha}{2}$, respectively. In particular, we ensure that first the nodes $v_1, ..., v_m$ are needed in cache for a computation; we can simply ensure this with a compute step for  node $u_1$ that has all of $v_1, ..., v_m$ as parents. Next, $v'$ is needed in cache; for this we compute another node $u_2$ (in a separate, second compute phase) that has $v'$ as its single parent. Finally, let $v_1, ..., v_m$ be again needed in cache in a third compute phase, due to computing a node $u_3$ that has all of $v_1, ..., v_m$ as parents. Let the memory weights of $u_1$, $u_2$, $u_3$ be $0$ (or negligibly small).

If there is a subset $A_0$ of sum $\frac{\alpha}{2}$, then we can leave this subset of nodes in cache after the computation of $u_1$ while we load $v'$. We then only have to load the remaining nodes from $v_1, ..., v_m$ (of total weight $\frac{\alpha}{2}$) back into cache for $u_3$; this gives a total loading cost of $\alpha + \frac{\alpha}{2} + \frac{\alpha}{2} = 2\alpha$. If such a subset does not exist, then the subset of $\{ v_1, ..., v_n\}$ that remains in cache when computing $u_2$ has size strictly less than $\frac{\alpha}{2}$, so the loading cost for the third computation is more than $\frac{\alpha}{2}$. This results in a total cost of more than $2 \alpha$.
\end{proof}

\renewcommand{\proofname}{Proof of Lemma~\ref{lem:strongNP}}
\begin{proof}
Similarly to a proof in~\cite{BSP_DAG_opdas}, we consider a reduction from $3$-partition: given integers $A=\{a_1, ..., a_{3m}\}$ with sum $\alpha\cdot m$, we need to partition $A$ into $m$ disjoint subsets of sum $\alpha$ each. Let $r=\infty$, and consider nodes $v_1, ..., v_{3m}$ with memory weights $a_1, ..., a_{3m}$, respectively, which are all parents of a specific node $u$ that is computed on processor $1$ in superstep $(m+1)$. As such, the nodes $v_1, ..., v_{3m}$ all need to be loaded into the cache of processor $1$ sometime during the first $m$ loading phases $\Psi_{load}$. Furthermore, assume that processor $2$ needs to load a separate value of weight exactly $\alpha$ in each of the communication phases $1$, ..., $m$.

This means that in the first $m$ communication phases, we anyway have a loading cost of $\alpha$ due to processor $2$, summing up to $\alpha\cdot m$. In order to not increase the loading costs in any of these communication phases, we need to load values of total weight at most $\alpha$ to processor $1$ in each phase. In this case, a loading cost of $\alpha\cdot m$ can be achieved if and only if $A$ can be partitioned into $m$ subsets of sum $\alpha$ each, which completes the reduction.

It only remains to present a situation where processor $2$ needs to load a separate node of weight $\alpha$ in each communication phase $1$, ..., $m$. This can be achieved by setting $r=\alpha\cdot m$, adding a node $w$ of weight $\alpha\cdot (m-1)$ which is computed on processor $1$ in superstep $1$, and for each compute phase $i \in \{ 2, ..., m+1\}$, creating a separate node $u_i$ that has two parents: $w$, and a distinct source node of weight $\alpha$. This implies that $w$ has to be kept in the cache of processor $2$ through the whole schedule, hence processor $1$ can only load one new source node into  cache in each communication phase, and has to delete it after using it as an input in the subsequent computation. Note that we again assume that the memory weight of node $u$ and all the nodes $u_i$ are $0$.
\end{proof}

Note that the above proofs only require minimal modification for the case when the I/O costs of the $\Psi_{save}$ phases are also considered. In particular, in the reductions, only the sink nodes need to be ever saved into slow memory, so the saving costs are identical to the sum of their weights. Since the weights of the sink nodes are $0$, saving them does not affect the total cost, and hence the same proofs apply for the total I/O cost.

\subsection{Proofs of Lemmas~\ref{lem:async1} and \ref{lem:async2}} \label{app:sync}

We now analyze how different the synchronous and asynchronous optimal solutions can be from each other. We first begin with the case when we optimize for asynchronous cost, and show that this can be rather suboptimal in terms of synchronous cost.

\renewcommand{\proofname}{Proof of Lemma~\ref{lem:async1}}
\begin{proof}
Consider $P$ processors for some even integer $P$, $r=\infty$, and $g=0$ (or very small). For simplicity, let $P':=P/2$. For all $i \in [P']$, we create the following nodes in our graph: $u_{i,1}$, $v_{i,1}$, $u_{i,2}$, $v_{i,2}$, ..., $u_{i,P'}$, $v_{i,P'}$, such that for all $j \in [P'-1]$, there is an edge from both $u_{i,j}$ and $v_{i,j}$ to both $u_{i,j+1}$ and $v_{i,j+1}$. For compute weights, we select some large integer $Z$, and we set $\omega(u_{i,j})=\omega(v_{i,j})=Z$ if $i=j$, and we set $\omega(u_{i,j})=\omega(v_{i,j})=1$ if $i \neq j$.

Finally, we add a source node $s$, and draw an edge from $s$ to all $u_{i,1}$, $v_{i,1}$. This $s$ can be ignored in the rest of the analysis; it is loaded in the beginning to all processors at no cost. It only ensures that no other nodes are sources, and hence they all need to be computed.

The optimum schedule according to an asynchronous cost function is simple: we assign $u_{i,j}$ to processor $i$ and superstep $j$, and assign $v_{i,j}$ to processor $P'+i$ and superstep $j$. This provides a valid solution, and since all processor pairs have $P'-1$ supersteps of compute weight $1$ and a single superstep of compute weight $Z$, according to the asynchronous cost function, the cost is $Z+P'-1$.

However, according to a synchronous cost function, all supersteps contain a node of compute weight $Z$, and hence even with $L=0$, the total cost is $P' \cdot Z$. In contrast to this, assume we keep the assignments to processors the same, but change the assignment to supersteps: we move all the nodes of cost $Z$ into the same superstep, and the remaining nodes into the supersteps before and after. That is, for all $i \in [P']$ we place both $u_{i,j}$ and $u_{i,j}$ into superstep $P'+j-i$. This creates one superstep of cost $Z$, and $2P'-2$ supersteps of cost $1$. As such, the cost ratio between the two solutions is $\frac{P' \cdot Z}{Z+2P'-1}$, which approaches $P'=\frac{P}{2}$ as we increase $Z$ while keeping $P$ fixed.
\end{proof}

\renewcommand{\proofname}{Proof of Lemma~\ref{lem:async2}}
\begin{proof}
Let $P=5$, $r=\infty$, and $g=0$ (or very small). Consider nodes $u_1$, $u_2$ that both have edges to nodes $u_3$ and $u_4$. Let us set $\omega(u_1)=\omega(u_2)=Z-1$, and $\omega(u_3)=\omega(u_4)=2 \cdot Z$ for some large integer $Z$. We add further nodes $v_1$, $v_2$, $v_3$, $v_4$, with $v_1$ having an edge to all of $v_2$, $v_3$ and $v_4$. We set $\omega(v_1)=2 \cdot Z$ and $\omega(v_2)=\omega(v_3)=\omega(v_4)=Z-1$. Finally, we add another isolated node $w$ with $\omega(w)=Z-1$. We again add an artificial source $s$ with edges to $u_1$, $u_2$, $v_1$ and $w$.

In this DAG, the best schedule will always assign $u_1$, $u_2$ to some processors in the first superstep, and $u_3$ and $u_4$ to some processors in the second superstep. As for the rest of the nodes, we can assign $v_1$ and $w$ to the first superstep, and $v_2$, $v_3$, $v_4$ to the second superstep on separate processors; however, this means that there is a node of compute cost $2 \cdot Z$ in both supersteps, so the total cost is $4 \cdot Z$. On the other hand, we can just assign $w$ to the first superstep, $v_1$ to the second superstep, and $v_2$, $v_3$, $v_4$ to the third supersteps on separate processors; then the total cost is $(Z-1) + 2 \cdot Z + (Z-1) = 4 \cdot Z -2$. This latter cost is smaller, so this is the optimum cost. In particular, one of the optimal solutions is where we assign $w$ and $v_1$ to the same processor $p$ in supersteps $1$ and $2$, respectively.

In an asynchronous setting, the solution above still has a cost $4 \cdot Z - 2$, since this is the total computation cost on processor $p$. In contrast, the former solution discussed above ($w$ and $v_1$ both in the first superstep in different processors) only has a cost of $2 \cdot Z + (Z-1) = 3 \cdot Z - 1$. The factor of difference between the two solutions is then $\frac{4 \cdot Z - 2}{3 \cdot Z - 1}$, which approaches $\frac{4}{3}$ if we increase $Z$.
\end{proof}

\subsection{Proof of Lemma~\ref{lem:ilp}} \label{app:monotone}

Finally, we prove that even when the optimum to the ILP problem with a specific time limit $T$ has multiple empty steps, the solution can still be suboptimal.

\renewcommand{\proofname}{Proof of Lemma~\ref{lem:ilp}}
\begin{proof}
Consider the simplest case of $P\!=\!1$ processor with all weights being uniformly $1$, i.e.\ single-processor red-blue pebbling with computation costs; our proof already applies to this case. We will assume the base ILP representation that does not apply step merging; otherwise, the corresponding costs and the analysis need to be slightly adjusted.

Consider the following simple modification of the zipper gadget used in~\cite{pebbling_papp,MPP_opdas}. We take two chains $(u_1, ..., u_d)$ and $(u'_1, ..., u'_d)$ of length $d$ each. We then add a chain $(v_0, v_1, ..., v_{m})$ of length $(m+1)$. In all three chains, subsequent pairs of nodes are connected by a directed edge. Node $v_0$ has an edge from both $u_d$ and $u_d'$. For all $i \in [m]$, node $v_i$ has an edge from $u_d$ if $i$ is an odd number, and from $u'_d$ if $i$ is an even number. Finally, we add a single source node $w$ that has an edge to all other nodes, and consider $r=4$. Note that any reasonable schedule always keeps a red pebble on $w$, and hence the problem is equivalent to scheduling the rest of the DAG with $3$ red pebbles, with the modification that nodes $u_1$ and $u'_1$ can be recomputed without I/O steps at any point.

The shortest possible pebbling sequence in this DAG first loads $w$, then computes $u_1, ..., u_d$ in order (deleting the red pebble from all nodes but $u_d$), computes $u'_1, ..., u'_d$ in order (deleting the red pebble from all nodes but $u'_d$), then saves both $u_d$ and $u'_d$ to slow memory, and computes $v_0$. We can then compute $v_1$, but for this the red pebble from $u'_d$ needs to be deleted before. From here, for each $v_i$ with $i \geq 2$, we need to load an input value from slow memory ($u_d$ if $i$ is odd, $u'_d$ if $i$ is even), then we can compute $v_i$, and then delete the input that was just loaded. Finally, we save $v_{m}$ to slow memory. This is a pebbling sequence that consists of $2d+m+1$ compute and $4+(m-1)$ I/O steps, hence $s:=2m+2d+4$ steps altogether and a cost of $C_0:=2d+m+1 + (4+(m-1)) \cdot g$, both in the synchronous setting with $L=0$ and in the asynchronous setting.

Intuitively, the cost of this solution can be reduced by replacing one of the I/O steps with $d$ consecutive compute steps: instead of loading e.g.\ $u_d$, we can compute $u_1, ..., u_d$ again. This decreases the cost by $g-d$, but requires $(d-1)$ further steps in the sequence. If we select $g \geq d$, then this is indeed a solution of lower cost.

As such, consider an ILP representation of the problem with $T_0:=s+1$ steps. If $d \geq 3$, then this does not allow us to recompute $u_d$ or $u'_d$ at any point, so the optimal solution will be the one outlined above, consisting of $s$ steps; thus there can be an empty step in the ILP representation. In fact, if we increase $T_0$ to up to $s+(d-2)$, then there can be up to $(d-2)$ empty steps in an optimal solution. However, once we have $T=s+(d-1)$, we can recompute either $u_d$ or $u'_d$ once, thus enabling a solution of cost $C:=C_0-(g-d)$.
\end{proof}

\section{ILP-based scheduler} \label{app:ilps}

This section discusses the details of our ILP-based scheduling method.

\subsection{ILP representation}

Recall that our ILP representation of MBSP scheduling is centered around binary variables $\textsc{compute}_{p,v,t}$, $\textsc{save}_{p,v,t}$, $\textsc{load}_{p,v,t}$, $\textsc{hasred}_{p,v,t}$ and $\textsc{hasblue}_{v,t}$, defined for each node $v$, processor $p$ and discrete time step $t$. For convenience, when the indices $p$, $v$ and $t$ are considered or summed over the entire set of processors, nodes and time steps, respectively, we will leave out this domain from the notation; that is, we simply write $\sum_v$ instead of $\sum_{v \in V}$.

\subsubsection{Fundamental constraints.}

Recall that the fundamental linear constraints of the ILP are already summarized in Figure~\ref{fig:ilp}: 
\begin{itemize}[topsep=4pt, itemsep=3pt]
 \item constraint (1) ensures the validity of load steps,
 \item constraint (2) ensures the validity of save steps,
 \item constraint (3) ensures the validity of compute steps without step merging,
 \item constraint (4) controls the presence of red pebbles,
 \item constraint (5) controls the presence of blue pebbles,
 \item constraint (6) ensures that a processor only executes a specific operation in a step (when we do not use step merging),
 \item constraint (7) enforces the memory bound,
 \item constraint (8) handles the initial state for red pebbles,
 \item constraint (9) handles the initial state for blue pebbles,
 \item constraint (10) handles the terminal state for blue pebbles.
\end{itemize}
These constraints already ensure the proper behavior of a pebbling sequence.

Recall that the step merging optimization can allow us to reduce the number of required time steps significantly; however, this also comes with some changes to the rules above:
\begin{itemize}[topsep=4pt, itemsep=3pt]
 \item to ensure that a processor only executes a specific kind of operation in a step, we now use two extra binary variables $\textsc{compstep}_{p,t}$ and $\textsc{commstep}_{p,t}$. Then for all processors $p$ and time steps $t$, we have
 \begin{gather*}
  \sum_{v} \, \textsc{compute}_{p,v,t} \leq |V| \cdot \textsc{compstep}_{p,t} \\
  \sum_{v} \, \textsc{save}_{p,v,t} + \textsc{load}_{p,v,t} \leq 2 \cdot |V| \cdot \textsc{commstep}_{p,t} \, ,
  \end{gather*}
 and finally, we have
 \[ \textsc{compstep}_{p,t} +\textsc{commstep}_{p,t} \leq 1 \, ; \]
 \item for the validity of compute steps, for edges $(u, v) \in E$, processors $p$ and time steps $t$, we now instead have \[ \textsc{compute}_{p,v,t} \leq \textsc{hasred}_{p,u,t} + \textsc{compute}_{p,u,t} \, . \]
\end{itemize}

\subsubsection{Supersteps and cost functions.}

We note that in this ILP representation, the schedule is not directly organized into supersteps; instead, the superstep will be formed from the consecutive blocks of compute steps and I/O steps in the sequence.

For the asynchronous version of this problem, we do not even need to consider the supersteps. In this case, the role of the function $\gamma$ is fulfilled by continuous variables $\textsc{finishtime}_{p,t}$ for all processors $p$ and time steps $t$, and the role of $\Gamma$ is fulfilled by continuous variables $\textsc{getsblue}_{v}$ for all nodes $v$. We also define a large constant $M=P \cdot (\sum_v \omega(v) + g \cdot \mu(v))$.

Recall that for $\textsc{finishtime}_{p,t}$, we have the constraint
\begin{equation*}
\begin{split}
\textsc{finishtime}_{p,t} \geq \textsc{finishtime}_{p,(t-1)} \, + \\ + \, \sum_v \, \omega(v) \cdot \textsc{compute}_{p,v,t} + g \cdot \mu(v) \cdot (\textsc{save}_{p,v,t} + \textsc{load}_{p,v,t}) \, ,
\end{split}
\end{equation*}
and for the dependence on the slow memory, for all nodes $v$, processors $p$ and time steps $t$, we have
\begin{equation*}
 \textsc{getsblue}_{v} \geq \textsc{finishtime}_{p,t} - M \cdot (1 - \textsc{save}_{p,v,t}) \, , \\
\end{equation*}
as well as
\begin{gather*}
\textsc{finishtime}_{p,t} \geq \textsc{getsblue}_{v} \, + \\ + \,  g \cdot \sum_{u \in V} \, \mu(u) \cdot \textsc{load}_{p,u,t}  - M \cdot (1 - \textsc{load}_{p,v,t})  \,  \, ;
\end{gather*}
note that in contrast to the simplified version in Section~\ref{sec:full_ILP}, this latter constraint also works with step merging. These constraints motivate the ILP solver to only save each node $v$ to slow memory at most once, at the save operation with the earliest finishing time; the finishing time of this save step will be a lower bound on $\textsc{getsblue}_{v}$. When the value of $v$ is loaded to some processor $p$, then the finishing time of this operation can only start at the time $\textsc{getsblue}_{v}$; the sum in the right hand side is the length of this loading operation, so the condition indeed expresses its earliest finishing time.

Finally, acontinuous variable $\textsc{makespan}$ is used with the constraints $\textsc{finishtime}_{p,t} \leq \textsc{makespan}$ for all $p$ and $t$, and the objective is simply to minimize $\textsc{makespan}$.

In contrast, in the synchronous setting is more complex. Firstly, we add a binary variable $\textsc{compphase}_{t}$ for each time step $t$, and for each $t$, and require that
\[ \sum_v \, \sum_p \, \textsc{compute}_{p,v,t} \leq P \cdot |V| \cdot \textsc{compphase}_{t} \, . \]
Without step merging, we similarly define $\textsc{savephase}_{t}$, $\textsc{loadphase}_{t}$ for save and load operations; with step merging, we instead have a single binary variable $\textsc{commphase}_{t}$ and
\[ \sum_v \, \sum_p \, \textsc{save}_{p,v,t} +\textsc{load}_{p,v,t} \leq 2 \cdot P \cdot |V| \cdot \textsc{commphase}_{t} \, . \]
For all $t$, we then also set $\textsc{compphase}_{t} + \textsc{commphase}_{t} \leq 1$ or $\textsc{compphase}_{t} + \textsc{savephase}_{t} + \textsc{loadphase}_{t} \leq 1$.

We then use binary variables $\textsc{compends}_{t}$ to indicate the endpoints of compute phases. The process is similar with $\textsc{commends}_{t}$ in case of step merging, and $\textsc{saveends}_{t}$, $\textsc{loadends}_{t}$ without it.

We then use a continuous variable $\textsc{compuntil}_{p,t}$ for all processors $p$ and time steps $t$, for which we set
\begin{equation*}
\begin{split}
 \textsc{compuntil}_{p,t} \geq \textsc{compuntil}_{p,(t-1)} \, + \\
 + \, \sum_v \omega(v) \cdot \textsc{compute}_{p,v,t} - M \cdot \textsc{commends}_{t} \, .
\end{split}
\end{equation*}
This constraint ensures that the compute costs of the nodes are added up and increasing in the $\textsc{compuntil}_{p,t}$, but only until we have a step with $\textsc{commends}_{t}=1$ (exactly before the beginning of the next compute phase), where they can be set back to $0$, since the large constant $M$ definitely makes the right side of the inequality negative. The process is similar for I/O costs.

Finally, we have continuous variables $\textsc{compinduced}_{t}$, which ensure that these summed up costs are only added to the objective function in the steps where we have $\textsc{compends}_{t}=1$. Recall that for this, for all processors $p$, we have
\[
 \textsc{compinduced}_{t} \geq \textsc{compuntil}_{p,t} - M \cdot (1 - \textsc{compends}_{t}) \, , \]
which also takes the maximum of the compute costs on the processors at the same time. The maximal I/O costs are again obtained in an identical way with a variable $\textsc{comminduced}_{t}$. Finally, the objective of the ILP problem is to minimize the expression
\[ \sum_t \, \textsc{compinduced}_{t} + \textsc{comminduced}_{t} + L \cdot \textsc{commends}_{t} \, . \]

The above conditions are slightly different when the indices $(t-1)$ or $(t+1)$ do not exist, and hence some variables are left out. 

\subsubsection{Further discussion.}

We also note that the values of some of our binary variables is already pre-determined in the problem; for instance, for a source node $v$, we always have $\textsc{hasblue}_{v,t}=0$ and $\textsc{compute}_{p,v,t}=0$. Due to the interface of the solver where it is easier to create variables in arrays, it is more convenient to still have these variables created. Then as a simpler solution, one could add an extra constraint to fix these variables to the appropriate value; with this, the preprocessing phase of a sophisticated ILP solver can remove the corresponding variables from all their constraints. Nonetheless, we have found that instead, it improves efficiency to edit the corresponding linear constraints manually, and ensure that these variables are never added at all.

Recall that in order to improve the efficiency of the solving process, we always initialize the ILP solver with an initial MBSP schedule that is found with the two-step approach. In our ILP, we set the number of steps $T$ slightly higher than the number of (merged) time steps in the ILP representation of this initial solution.

\subsection{Divide-and-conquer algorithm} \label{app:dnq}

Recall that we also introduce an approach to handle DAGs of larger size (e.g. a few hundred nodes) by splitting them into smaller scheduling subproblems. This divide-and-conquer approach consist of several different steps.

Firstly, we want to partition the input DAG into smaller parts such that the quotient graph is acyclic (often called an `acyclic partitioning' of the DAG). We can then develop a schedule for each part independently, and then concatenate these into a schedule for the whole DAG. We also want to ensure that the parts are relatively disjoint from each other, with only a few edges between them.

This partitioning problem can also be formulated as an ILP, where the number of cut hyperedges (which is a known to be an accurate indicator of the amount of communicated data~\cite{papp2023partitioning}) is in the objective function, and the acyclicity condition is expressed with linear constraints. The size of this ILP representation is significantly smaller than the ILP of the scheduling problem on the same DAG, since the partitioning problem does not include a time dimension in the variables. Our experience shows that satisfying the acyclicity constraint for more than $2$ parts is significantly more challenging for the COPT solver than for $2$ parts; in fact, when we only have $2$ parts (known as bipartitioning), COPT has almost always found the optimal acyclic partitioning of any of our DAGs within a second.

As our first step, we use the ILP-based acyclic bipartitioning method outlined above, and we apply this recursively to split the DAG into smaller subDAGs of the desired size. In particular, in each step, we consider one of our subDAGs that has $n_0$ nodes, and if we still have $n_0 > 60$, then we further bipartition it acyclically into two parts with at least $n_0 / 3$ nodes in each.

Then as a second step, we create a scheduling plan for the quotient graph with an adjusted version of the BSPg scheduling heuristic~\cite{BSP_algos_opdas}. In the quotient graph, we simply sum the weights $\omega$ and $\mu$ to assign weights to each contracted node. We can then naturally modify the BSPg heuristic such that it allows to keep assigning several processors to a node, with the corresponding execution time reduced proportionally. The resulting high-level schedule tells us which set of processors should be used for each partition, i.e.\ each scheduling subproblem.

As discussed in Section~\ref{sec:dnq}, in the third step, we use our ILP-based scheduler to find a good schedule for each scheduling subproblem, with the appropriate modifications to the ILP representation.

As a final step, we concatenate the subDAG schedules into an MBSP schedule for the entire problem. For each concrete subproblem, the superstep indices are now adjusted to begin at the maximal value of the current superstep indices of the processors used in the subproblem. After this combined MBSP schedule is created, there are a few technical steps to streamline and further improve this schedule. For instance, we check if some of the consecutive compute phases can be merged into a single superstep, which may happen on the border of two consecutive subschedules. We also check whether we can remove some delete steps that were artificially inserted at the border of two subschedules. For instance, if a node $v$ only appears as a parent node in the first and third subDAGs, then it is deleted from cache before the second subschedule (since $v$ is not even present in the corresponding ILP problem); however, if the second subproblem does not utilize the cache to its full capacity, we can potentially avoid this deletion of $v$ to forego a load operation of $v$ in the third subDAG.

Recall that our experiments show that this approach is only able to provide an improvement over the two-stage baseline on specific DAGs; on others, it can actually return a worse MBSP schedule than the baseline. This is due to the fact that the ILPs now only capture a specific part of the problem, and the cost functions they optimize is not the global optimum. It is a promising direction for future work to gain a deeper understanding of this phenomenon, and to identify specific applications or families of DAGs where our divide-and-conquer approach can reliably outperform the baseline.

\section{Experiment details} \label{app:experiments}

Finally, we discuss details of the experiments and empirical results.

\subsection{Experimental setup}

As our main dataset, we use the smallest dataset of computational DAGs from~\cite{BSP_algos_opdas} called `tiny' in their paper. This contains $3$ coarse-grained DAGs, and 3-3 fine-grained representations of different-size instances of CG, SpMV, iterated SpMV and $k$-NN. We refer the reader to~\cite{BSP_algos_opdas} for more details.

We also use a small sample of $10$ DAGs from the next dataset (`small') of~\cite{BSP_algos_opdas}. $9$ out of these have between $264$ and $322$ nodes, while one of the SpMV DAGs has $464$ nodes.

Since the DAGs have no memory weights, we assigned uniformly and independently at random a weight $\mu(v) \! \in \! \{1, 2, 3, 4, 5\}$ to each node; this is in the same magnitude as the compute weights, so with a choice of $g\!=\!1$, this means that both the computation and I/O costs play an important role in the problem.

We use the COPT ILP solver to find good solutions to our scheduling (sub)problems. While the solving process in COPT is highly customizable, we do not explore this further, and leave most of the parameters of COPT at their suggested default value.

As a baseline scheduler in the two-stage approach, we use the BSPg scheduling heuristic from~\cite{BSP_algos_opdas}. We also consider the Cilk work-stealing scheduler for this task; a version of this scheduler adapted to the BSP model is also included in our OneStopParallel framework. For memory management, the clairvoyant algorithm simply considers the following compute steps on the same processor to establish for each value in cache when this value is next required in the future. When having to evict a value from cache, it always evicts the value that is not required for the longest time. Note that in our implementation, values that are not required anymore in the future are always evicted from cache automatically. As an alternative memory management approach, we consider the `least recently used' method, which maintains for each value in cache the last time they were active (computed or used as an input), and whenever a value needs to be evicted, it selects the value that was active the longest time ago.

Besides our algorithms, the OneStopParallel framework also provides a test suite to run our experiments, which outputs the costs of the obtained MBSP schedules into a .csv file.

\begin{table*}[t]
    \renewcommand{\arraystretch}{1.3}
	\caption{Cost of MBSP schedules on each instance of our main dataset. The $5$ columns respectively correspond to: 1) our main baseline (BSPg + Clairvoyant algorithm), 2) our ILP-based MBSP scheduler initialized with this main baseline, 3) the weaker baseline of Cilk + LRU, 4) the stronger baseline of an ILP-based BSP scheduler + the Clairvoyant algorithm, and 5) our ILP-based MBSP scheduler initialized with this stronger baseline. Note that the first two columns correspond to the values in Table~\ref{tab:tiny_dataset}.} \label{tab:main_exps}
	\centering
	\begin{tabular}{ c || c | c || c || c | c ||}
		$\:$ Instance $\:$ & $\:$ Baseline $\:$ & $\:$ \makecell{Our ILP \\ for MBSP} $\:$ & $\,$ \makecell{Weak base \\ (Cilk+LRU)} $\,$ & $\:$ \makecell{BSP ILP \\ baseline} $\:$ & $\:$ \makecell{BSP ILP \\ + our ILP} $\:$ \\ \hline \hline
		bicgstab & 197 & 181 & 212 & 135 & 122 \\ \hline
        $k$-means & 158 & 106 & 163 & 100 & 98 \\ \hline
        pregel & 206 & 152 & 210 & 160 & 145 \\ \hline
        spmv\_N6 & 123 & 79 & 166 & 92 & 79 \\ \hline
        spmv\_N7 & 120 & 77 & 138 & 92 & 75 \\ \hline
        spmv\_N10 & 159 & 96 & 190 & 111 & 94 \\ \hline
        CG\_N2\_K2 & 283 & 267 & 310 & 214 & 194 \\ \hline
        CG\_N3\_K1 & 199 & 195 & 263 & 287 & 281 \\ \hline
        CG\_N4\_K1 & 229 & 208 & 268 & 324 & 314 \\ \hline
        exp\_N4\_K2 & 149 & 91 & 152 & 104 & 90 \\ \hline
        exp\_N5\_K3 & 185 & 144 & 251 & 214 & 147 \\ \hline
        exp\_N6\_K4 & 169 & 168 & 225 & 210 & 200 \\ \hline
        kNN\_N4\_K3 & 179 & 132 & 170 & 132 & 108 \\ \hline
        kNN\_N5\_K3 & 167 & 108 & 192 & 144 & 108 \\ \hline
        kNN\_N6\_K4 & 180 & 173 & 241 & 181 & 178 \\ \hline
	\end{tabular}
\end{table*}

\subsection{Experimental results}

Recall that for our main experiment, we use $P=4$ processors, the synchronous model with $L=10$, and a memory bound of $r=3 \cdot r_0$. For completeness, we also present the costs of the concrete MBSP schedules for each instance in Table~\ref{tab:main_exps}. We also include the practical baseline (Cilk + LRU) in this table, as well as the stronger baseline with the ILP-based BSP scheduling.

The corresponding geomean factor improvements in cost are listed in the main part of the paper. The results indicate that the ILP-based scheduler can indeed improve significantly on the main baseline. The factor of improvement also varies significantly between the concrete computational DAGs: the smallest improvement is on exp\_N6\_K4, where the cost is only reduced by $1$, whereas the largest difference is on spmv\_N10, and it amounts to a $0.60\times$ factor.

When considering the two-stage baseline with the Cilk scheduler and the LRU eviction policy, which may be more representative of actual applications, we see that this returns weaker schedules than our main baseline, and hence the improvements are even larger with respect to these schedules. The smallest improvement with respect to this baseline is observed on CG\_N2\_K2, and amounts to $0.86\times$. The largest improvement here is on the DAG spmv\_N6, and it amounts to a $0.48\times$ factor.

We also see that in general, the two-stage approach with the ILP-based BSP scheduler mostly returns stronger baseline schedules, and from these initial solutions, our MBSP ILPs can often find even better schedules in the end. Note that there are also some case where the ILP-based baseline returns a weaker solution than our main baseline. While counter-intuitive, this can indeed happen: the BSP-based ILP scheduler is in fact optimizing for an inappropriate cost function by ignoring the memory limitations. In these cases, our own ILP-scheduler also ends up with a weaker MBSP schedule in the end, since it is launched with a weaker initial solution. 

For the remaining experiments, we only consider the main baseline and the ILP method, and hence we show them in the same column, separated by a `$/$' sign. Table~\ref{tab:main_other} contains the schedule costs for some other variants of our experiments: with a memory bound of $r=5 \! \cdot \! r_0$ or $r= r_0$, with $P=8$ processors instead of $4$, with $L=0$, and with the asynchronous cost function.

\begin{table*}[t]
    \renewcommand{\arraystretch}{1.3}
    \caption{Cost of MBSP schedules with the baselines / our ILP methods in alternative cases: 1) with $r\!=\!5 \! \cdot \! r_0$, 2) $r\!=\!r_0$, 3) with $P\!=\!8$, 4) with $L\!=\!0$, 5) with the asynchronous cost function.} \label{tab:main_other}
	\centering
	\begin{tabular}{ c || c | c | c | c | c ||}
		$\:$ Instance $\:$ & $\:\,$ $r=5 \! \cdot \! r_0$ $\:\,$ & $\;\;\;$ $r= r_0$ $\;\;\;$ & $\;\;\;$ $P=8$ $\;\;\;$ & $\;\;\;$ $L=0$ $\;\;\;$ & $\;\;\;$ async $\;\;\;$ \\ \hline \hline
		bicgstab & 197 / 146 & 221 / 213 & 176 / 173 & 117 / 89 & 92 / 83 \\ \hline
        $k$-means & 158 / 124 & 176 / 173 & 156 / 102 & 88 / 74 & 75 / 68 \\ \hline
        pregel & 206 / 148 & 222 / 222 & 160 / 138 & 146 / 142 & 135 / 118 \\ \hline
        spmv\_N6 & 123 / 79 & 167 / 116 & 104 / 75 & 83 / 55 & 70 / 54 \\ \hline
        spmv\_N7 & 120 / 75 & 134 / 132 & 83 / 68 & 80 / 55 & 66 / 50 \\ \hline
        spmv\_N10 & 159 / 96 & 215 / 215 & 124 / 69 & 119 / 80 & 104 / 79 \\ \hline
        CG\_N2\_K2 & 283 / 193 & 366 / 366 & 295 / 291 & 163 / 152 & 133 / 133 \\ \hline
        CG\_N3\_K1 & 199 / 194 & 343 / 341 & 176 / 176 & 129 / 116 & 112 / 107 \\ \hline
        CG\_N4\_K1 & 229 / 219 & 343 / 343 & 205 / 202 & 159 / 151 & 122 / 122 \\ \hline
        exp\_N4\_K2 & 149 / 95 & 201 / 195 & 138 / 84 & 89 / 80 & 71 / 67 \\ \hline
        exp\_N5\_K3 & 185 / 166 & 261 / 261 & 185 / 182 & 115 / 110 & 89 / 89 \\ \hline
        exp\_N6\_K4 & 169 / 167 & 257 / 254 & 165 / 165 & 99 / 97 & 83 / 80 \\ \hline
        $\,$ kNN\_N4\_K3 $\,$ & 179 / 110 & 242 / 242 & 143 / 105 & 109 / 95 & 78 / 76 \\ \hline
        $\,$ kNN\_N5\_K3 $\,$ & 167 / 120 & 213 / 212 & 162 / 101 & 107 / 94 & 86 / 84 \\ \hline
        $\,$ kNN\_N6\_K4 $\,$ & 180 / 178 & 302 / 297 & 190 / 190 & 120 / 111 & 87 / 87 \\ \hline
	\end{tabular}
\end{table*}

The data allows us to investigate how the different choices of the memory bound parameter affect the cost of the schedules. The table shows that going from $r=3 \! \cdot \! r_0$ to $r=5 \! \cdot \! r_0$ does not affect the baseline, which suggests that this small change in the memory bound has little effect on the problem as a whole. However, this still implies that the ILP solver will make different decisions when exploring the search space, and hence it often ends up with a significantly different MBSP schedule. While the final cost is smaller in most cases with $r=5 \! \cdot \! r_0$, there are also a few instances where the solver ends up with a slightly worse solution, even though the memory bound is looser.

In contrast to this, when the memory bound is set to the minimal value of $r=r_0$, the costs of the schedules notably increase, and the improvements to the baseline also drop, since this setting allows for less freedom when designing our schedules. The largest increase to the baseline in this setting is of a $0.69\times$ factor on spmv\_N6; for all other instances, the value is above $0.96\times$.

When switching to $P\!=\!8$ processors, we usually see a reduction in the baseline cost, since there are now more opportunities to parallelize the computation. However, the effects on the ILP cost are more ambiguous, since this also comes with almost twice as many variables in the ILP representation, and hence the task of the ILP solver becomes significantly more challenging.

When changing $L\!=\!10$ to $L\!=\!0$, the costs of the schedules drop significantly. The difference between the ILP and the baseline also becomes slimmer in this case, since one of the main strengths of the ILP approach is that it also directly considers and minimizes synchronization costs, in contrast to the baseline.

Finally, in the asynchronous case, we again see schedules of significantly lower cost. This setting is much more challenging for our ILP-based solvers due to the interdependences between the finishing time variables. We see that the improvement on the SpMV instances still remains $0.76\times$-$0.77\times$, but on many other instances, the solver cannot improve on the initial solution at all.

Note that we also ran some experiments with a choice of $P\!=\!1$, which is essentially the single-processor red-blue pebble game of Hong and Kung~\cite{hongkung} extended with computation costs and node weights. In this case, our baseline corresponds to a DFS ordering combined with the clairvoyant cache eviction strategy. We found that this is a rather strong baseline, and our schedulers could almost never improve upon it. In particular, for $r=3 \! \cdot \! r_0$, the ILP only improved upon the baseline for two instances: $0.89\times$ for exp\_N4\_K2, and $0.63\times$ for exp\_N5\_K3. In case of $r=r_0$, the ILP could not improve on the initial schedule at all for any of the instances.

We also briefly considered the ILP-based method without recomputation, i.e.\ when each node can only appear in a single compute step over all processors and time steps. The BSPg scheduler also fulfills this property, so the setting still allows us to use this baseline for initializing the ILP without any changes; we only restrict the search space for the MBSP schedules considered by the ILP solver. This resulted in schedules of higher cost for $7$ out of the $15$ instances, with the largest increase of $1.40\times$ obtained on kNN\_N5\_K3. However, there were also $6$ instances where counter-intuitively, this resulted in schedules of smaller cost in the end, since the extra constraints can also be a major help for the ILP solver when optimizing the solution within a fixed time limit.

Finally, we briefly consider our approach on $10$ DAGs from the larger dataset. Due to the much larger size of these instances, we used the divide-and-conquer ILP here. We also used $r=5 \!\cdot \! r_0$ to accommodate these larger DAGs where we may have much more intermediate steps before a given value in cache is reused.

The cost of the schedules in each instance is summarized in Table~\ref{tab:small_dataset}. One can observe that the ILP can find significantly better schedules on some instances: on the coarse-grained instances simple\_pagerank and snni\_graphchallenge, the improvements are $0.77\times$ and $0.60\times$, respectively, and on the SpMV instances spmv\_N25 and spmv\_N35, the improvements are $0.74\times$ and $0.76\times$, respectively. The cost also improves by $0.89\times$ on the instance CG\_N5\_K4, and is identical on CG\_N7\_K2. However, on the remaining instances, the ILP-based schedule is actually worse than the baseline, with a geomean increase of $1.24\times$ and a largest increase of $1.29\times$ on exp\_N15\_K4. This effect is discussed in more detail in Section~\ref{app:dnq}.

\end{document}